\documentclass[review]{elsarticle}
\usepackage{amssymb,amsmath}
\usepackage{amsmath}
\usepackage{amsfonts,url}
\usepackage{amsmath,cancel}
\usepackage{amsfonts,cancel}

\newtheorem{proof}{Proof}

\usepackage{amssymb}
\usepackage[dvips,final]{epsfig}
\usepackage{amsfonts}
\usepackage{latexsym}

\usepackage{subfigure}
\usepackage{color}

\usepackage{setspace}

\usepackage{multirow}

\newcommand{\kk}{\mbox{$\cal K$}}

\newcommand{\oo}{\mbox{$\mathbb O$}}

\newcommand{\uu}{\mbox{$\mathbf u$}}
\newcommand{\x}{\mbox{$\mathbf x$}}
\newcommand{\y}{\mbox{$\mathbf y$}}
\newcommand{\vv}{\mbox{$\mathbf v$}}
\newcommand{\w}{\mbox{$\mathbf w$}}
\newcommand{\zz}{\mbox{$\mathbf z$}}
\newcommand{\s}{\mbox{$\mathbf s$}}
\newcommand{\gf}{\mbox{$\mathbf g$}}
\newcommand{\h}{\mbox{$\mathbf h$}}

\newcommand{\qq}{\mbox{$\mathbf q$}}

\newcommand{\ttt}{\mbox{$\cal T$}}
\newcommand{\f}{\mbox{$\cal F$}}

\newcommand{\p}{\mbox{$\cal P$}}

\newcommand{\sss}{\mbox{$\cal S$}}
\newcommand{\rrr}{\mbox{$\cal R$}}

\newcommand{\qa}{\mbox{\quad\mbox{and}\quad}}

\newcommand{\rt}{\mbox{$\mathbb R$}}

 \newcommand{\rank}{\mathrm{rank\;}}

\newcommand{\bxi}{\mbox{\boldmath $\xi$}}

 \def\diag{\mathop{{\rm diag}}\nolimits}

\newtheorem{theorem}{Theorem}

\newtheorem{proposition}{Proposition}

\newtheorem{example}{Example}

\newtheorem{remark}{Remark}

\journal{XXX}

\begin{document}

\begin{frontmatter}

\title{Extended Principal Component Analysis}

\author[mymainaddress,mysecondaryaddress]{Pablo Soto-Quiros}
\ead{juan.soto-quiros@mymail.unisa.edu.au}

\author[mymainaddress]{Anatoli~Torokhti\corref{mycorrespondingauthor}}
\cortext[mycorrespondingauthor]{Corresponding author}
\ead{anatoli.torokhti@unisa.edu.au}

\address[mymainaddress]{School of Information Technology and Mathematical Sciences, University of South Australia, SA 5095, Australia}
\address[mysecondaryaddress]{Instituto Tecnologico de Costa Rica, Apdo. 159-7050, Cartago, Costa Rica}

%\author{ Pablo~Soto-Quiros$^{1,2}$ and  Anatoli~Torokhti$^{1}$}

%\address{$^{1}$School of Inf. Techn. \& Math. Sci., University of South  Australia, SA 5095, Australia.\\
%  anatoli.torokhti@unisa.edu.au,  juan.soto-quiros@mymail.unisa.edu.au, \\
%  $^{2}$Instituto Tecnologico de Costa Rica, Apdo. 159-7050, Cartago, Costa Rica  }

\begin{abstract}
Principal Component Analysis (PCA)  is a transform  for finding the principal components (PCs) that represent features of random data. PCA also provides a  reconstruction of the PCs to the original data.
We consider an extension of PCA which allows us to improve the associated accuracy and diminish the numerical load, in comparison with  known techniques.  This is achieved due to the  special structure of the proposed transform which contains two  matrices $T_0$ and $T_1$, and a special transformation $\f$ of the so called auxiliary random vector $\w$.   For this reason, we call it the three-term PCA. In particular, we show that the three-term PCA  always exists, i.e. is applicable to the case of singular data.
Both rigorous theoretical justification of the three-term PCA and simulations with real-world data are provided.

\end{abstract}

\begin{keyword} Principal component analysis \sep Least squares linear estimate \sep  Matrix computation \sep  Singular value decomposition.
\end{keyword}

\end{frontmatter}

%\section{Introduction}\label{intro}

\section{Introduction}\label{mot}

In this paper, an extension  of Principal Component Analysis (PCA) and its rigorous justification are considered. In comparison with known techniques, the proposed extension of PCA allows us to improve the associated accuracy and diminish the numerical load. The innovation of the proposed methodology, differences from the known results and advantages  are specified in Sections \ref{wn9w},  \ref{dif77},  \ref{x78an} and \ref{x78kl}.

 PCA is a technique for finding so called principal components (PCs) of the data of interest represented by a large random vector, i.e. components of a smaller vector which preserve the principal features of the data. PCA also provides a  reconstruction of PCs to the original data vector  with the least possible error among all linear transforms. According to Jolliffe \cite{Jolliffe2002}, ``Principal component analysis is probably the oldest and best known of the techniques of multivariate analysis''.
 This is a topic of  intensive  research which has  an enormous number of related references. For instance,  a Google search for `principal component analysis' returns about $8,160,000$ results. In particular, the  references which are most related to this paper  (from our point view) are represented in \cite{Brillinger2001,681430,905856,Tomasz1293,DuongNguyen2014288,tor5277,Scharf1991113}. PCA is used in a number  of application areas (in an addition to the previous references see, for example, \cite{diam96,du2006,Saghri2010,gao111}). Therefore, related techniques with a better performance are of vital importance.

By PCA in \cite{Jolliffe2002,Brillinger2001},  under the strong restriction of invertibility of the  covariance matrix, the procedures for finding the PCs and their reconstruction  are determined by  the matrix of the rank less or equal to $k$ where $k$ is the number of required PCs.
For a fixed number of  PCs, the PCA accuracy of their reconstruction to the original data cannot be improved. In other words, PCA has the {\em only degree of freedom} to control the accuracy, it is the number of PCs. Moreover, PCA in the form obtained, in particular, in \cite{Jolliffe2002,Brillinger2001}, is not applicable if the associated covariance matrix is singular. Therefore,  in \cite{tor843}, the generalizations of PCA,  called generalized Brillinger transforms (GBT1 and GBT2), have been developed. First, the GBT1 and GBT2 are applicable to the case of singular covariance matrices. Second, the GBT2 allows us to improve the errors associated with PCA and the GBT1. We call it the generic Karhunen-Lo\`{e}ve transform (GKLT).
The GKLT  requires  an additional knowledge of covariance matrices $E_{x y^2}$ and $E_{y^2 y^2}$ where $\x$ and $\y$ are stochastic vectors, and $\y^2$ is a Hadamard square of $\y$  (more details are provided in Section \ref{gbt1} below). Such  knowledge may be difficult. Another difficulty associated with the GBT2 and GKLT is their  numerical load which is larger than that of PCA.  Further, it follows from \cite{924074} that the GKLT accuracy is better than that of PCA in \cite{Jolliffe2002,Brillinger2001,681430} subject to the condition which is quite difficult to verify (see Corollary 3 in \cite{924074}). Moreover, for the GBT2 in \cite{tor843}, such an analysis has not been provided.

We are motivated by the development of an extension of PCA  which covers the above drawbacks.
%The main purpose of this work is to show ...
We provide both a detailed theoretical analysis of the proposed extension of PCA and numerical examples that illustrate the theoretical results.

%The review of the GBT1, GBT2 and GKLT is given below  in Section \ref{app1} `Appendix'.

\section{Review of PCA and its known generalizations GB1, GB2, GKLT}\label{gbt1}

First, we introduce some notation which is used  below.

Let $\x=[\x_1,\ldots, \x_m]^T\in  L^2(\Omega,\mathbb{R}^m)$ and $\y=[\y_1,\ldots, \y_n]^T\in  L^2(\Omega,\mathbb{R}^n)$  be  random vectors\footnote{Here, $\Omega=\{\omega\}$ is the set of outcomes, $\Sigma$ a $\sigma$-field of measurable subsets of $\Omega$, $\mu:\Sigma\rightarrow[0,1]$ an associated probability measure on $\Sigma$ with $\mu(\Omega)=1$ and $(\Omega,\Sigma,\mu)$ for a probability space.}. Vectors $\x$ and $\y$ are interpreted as  reference data and  observable data, respectively, i.e.  $\y$ is a noisy version of $\x$. Dimensions $m$ and $n$ are assumed to be large.
  Suppose we wish to denoise  $\y$ and reduce it   to a `shorter'  vector $\uu\in  L^2(\Omega,\mathbb{R}^k)$ where $k \leq \min \{m, n\}$, and then  reconstruct  $\uu$ to vector $\widetilde{\x}$ such that $\widetilde{\x}$ is as close to $\x$ as possible. Entries  of vector $\uu$ are called the principal components (abbreviated above as PCs).

Let us write
$
\displaystyle \|{\bf x}\|^2_{\Omega} =\int_\Omega \|{\bf x}(\omega)\|_2^2 d\mu(\omega) < \infty,
 $
where $\|{\bf x}(\omega)\|_2$ is the Euclidean norm of ${\bf x}(\omega)\in\mathbb{R}^m$.
Throughout the paper, we assume that means $E[\x]$ and $E[\y]$ are known. Therefore, without loss of generality, we will assume henceforth that  $\x$ and $\y$ have zero means. Then the covariance matrix formed from $\x$ and  $\y$ is given by $E_{xy}=E[\x\y^T]=\{e_{ij}\}_{i,j=1}^{m,n}\in\rt^{m\times n}$ where $\displaystyle e_{ij} = \int_\Omega \x_i(\omega) \y_j(\omega)  d\mu(\omega)$.

Further, the singular value decomposition (SVD) of matrix  $M\in \rt^{m\times n}$ is given by  $M=U_M\Sigma_M V_M^T$ where  $U_M=[u_1 \;u_2\;\ldots u_m]\in \rt^{m\times m}, V_M=[v_1 \;v_2\;\ldots v_n]\in \rt^{n\times n}$  are unitary matrices, and  $\Sigma_M=\diag(\sigma_1(M),$ $\ldots,$ $\sigma_{\min(m,n)}(M))\in\rt^{m\times n}$
 is a generalized diagonal matrix,  with the singular values $\sigma_1(M)\ge \sigma_2(M)\ge\ldots\ge 0$ on the main diagonal.
Further, $M^{\dag}$ denotes the Moore-Penrose pseudo-inverse matrix for matrix $M$.

The generalizations of PCA mentioned in Section \ref{mot}, GBT1 and GBT2  \cite{tor843}, are represented as follows. Let us first consider the GBT2  since the GBT1 is a particular case of the GBT2. The GBT2 is given by
 \begin{eqnarray}\label{bnm11}
B_2(\y)=R_1 [P_1\y + P_2\vv],
 \end{eqnarray}
where $\vv\in L^2(\Omega,\mathbb{R}^n)$ is an `auxiliary' random vector used to further optimize the transform, {\bf (Assoc. Edit.: More explanations!!!) } and matrices  $R_1\in\rt^{m\times k}$, $P_1\in\rt^{k\times n}$ and $P_2\in\rt^{k\times n}$ solve the problem\footnote{Note that in (\ref{bnm11}) and (\ref{byfx}), strictly speaking, $R_1$ $P_1$ and $P_2$ should be replaced with  operators $\rrr_1: L^2(\Omega,\mathbb{R}^k) \rightarrow L^2(\Omega,\mathbb{R}^m)$, $\p_1: L^2(\Omega,\mathbb{R}^n) \rightarrow L^2(\Omega,\mathbb{R}^k)$ and $\p_2: L^2(\Omega,\mathbb{R}^n) \rightarrow L^2(\Omega,\mathbb{R}^k)$, respectively.  This is because  each matrix, say, $R_1\in\rt^{m\times k}$ defines a bounded linear transformation $\rrr_1: L^2(\Omega,\mathbb{R}^k) \rightarrow L^2(\Omega,\mathbb{R}^m)$. Nevertheless, it is customary to write $R_1$  rather then $\rrr_1$, since  $[\rrr_1(\uu)](\omega) = R_1[\uu(\omega)]$, for each $\omega\in \Omega$. We keep this type of notation throughout the paper.}
\begin{equation}\label{byfx}
\min_{R,\hspace*{1mm} [P_1 P_2]} \displaystyle \|\x - R_1 [P_1\y + P_2\vv]\|_\Omega^2
\end{equation}
so that, for  $\qq=[\y^T \vv^T]^T\in   L^2(\Omega,\mathbb{R}^{2n})$ and $G_q = E_{x q}E_{qq}^{\dag}E_{q x}$,
  \begin{eqnarray}\label{ac11}
  R_1= U_{G_q,k},\hspace*{4mm}
  [P_1, P_2] = \hspace*{-1mm}U_{G_q,k}^T E_{x q}E_{qq}^{\dag},
\end{eqnarray}
where  $U_{G_q,k}$ is formed by the first $k$ columns of $U_{G_k}$, and $P_1$ and $P_2$ are represented by the corresponding blocks of matrix $U_{G_q,k}^T E_{x q}E_{qq}^{\dag}$.
The principal components are then given by $\uu = [P_1, P_2][\y^T \vv^T]^T$.

The GBT1 follows from the GBT2 if $R_1 P_2\vv=\mathbf 0$, where $\mathbf 0$ is the zero vector. That is, the GBT2 has one matrix more (i.e., one degree of freedom more) than the GBT1. This allows us to improve the GBT2 performance compared to that by the GBT1 (see \cite{tor843} for more detail).

The GKLT \cite{924074} is given by
 \begin{eqnarray}\label{qmiw8}
\kk (\y)=K_1\y + K_2\y^2,
 \end{eqnarray}
where matrices $K_1\in\rt^{m\times n}$ and $K_2\in\rt^{m\times n}$ solve the problem
\begin{equation}\label{by29}
\min_{ \substack{[K_1 K_2]\\ \rank [K_1 K_2]\leq k}} \displaystyle \|\x - [K_1 \y + K_2\y^2]\|_\Omega^2,
\end{equation}
and $\y^2$ is given by the Hadamard square  so that $\y^2(\omega) = [\y_1^2(\omega),\ldots,\y_n^2(\omega)]^T$, for all $\omega\in\Omega$.

PCA is a particular case of the GKLT if $K_2\y^2={\mathbf 0}$ and matrix $E_{yy}$ is non-singular.
The PCA, BT, GBT1, GKLT and GBT2    follow from the solution of essentially the same optimization problem. The differences are that first, the associated solutions are obtained under different assumptions  and second,  the solutions result in transforms that have  different computational schemes. More details can be found in \cite{tor843}.
Further, each of PCA and the GBT1  has $m\times n$ parameters to optimize, which are entries of matrices $K_1$ and $R_1P_1$, respectively.   Similar to PCA, the GBT1 has  one degree of freedom to improve the associated accuracy,  it is the number $k$ of PCs, i.e., the dimension  of vector $\uu$. Thus, for fixed $k$, the GBT1 accuracy cannot be improved.
 %The second degree transform \cite{924074}  is a generalization of the GKLT and follows from the solution of a {\em single} optimization problem as well but has twice more parameters to optimize.
 The GKLT \cite{924074} and GBT2 \cite{tor843} each have two degrees of freedom, $k$ and  one matrix more than in PCA and GBT1. That is, the GKLT and GBT2 have twice as many parameters to optimize compared to PCA and GBT1. It is shown in \cite{tor843,924074} that this feature allows us to improve the accuracy associated with  the GKLT and GBT2.% More details are provided in \cite{tor843,924074}.

  % PCA, BT, GBT and GKLT   follow from the solution of the same, in fact, optimization problem. The differences are that first, the associated solutions are obtained under different assumptions  and second, the  PCA, BT, GBT and GKLT have  different computational schemes.

%\subsection{\bf\em{Contribution and novelty}}\label{wn9w}
\section{Contribution and novelty}\label{wn9w}

We propose and justify the PCA extension which\\
\hspace*{5mm} $\bullet$ always exists, i.e. is applicable to the case of singular data (this is because it is constructed in terms of pseudo-inverse matrices; see Section \ref{det}),\\
\hspace*{5mm} $\bullet$ has better associated accuracy than that of the GBT1, GBT2 and GKLT (Sections \ref{298an}, \ref{1n8an}, \ref{nbm198}), \\
\hspace*{5mm} $\bullet$ has more degrees of freedom to improve the associated accuracy than the PCA, GBT1, GBT2 and GKLT (Sections  \ref{xvb8an} and  \ref{speccases}),\\
\hspace*{5mm} $\bullet$  has a lower computational load than that of the GBT2 and GKLT; in fact, for large $m, n$, it is  about 37\% of  that of the GBT2 and 22\% of that of the GKLT (Section \ref{cm9vn}),\\
\hspace*{5mm} $\bullet$ does not require the usage of matrices $E_{x y^2}$ and $E_{y^2 y^2}$ (as required in \cite{924074}) which are difficult to determine.

Further, we show, in particular,  that \\
\hspace*{5mm} $\bullet$ the  condition for the GKLT  \cite{924074} mentioned in Section \ref{gbt1} (under which the accuracy improvement is achieved) can be omitted (Section \ref{298an}).

In more detail, in  the proposed PCA extension, the additional degrees of freedom are provided by the auxiliary random vectors $\w$  and $\h$ which are introduced below in  Sections \ref{w0an} and  \ref{xvb8an}, respectively.  Vectors $\w$ and $\h$ are called  $\w$-injection and $\h$-injection. An improvement in the accuracy of the proposed transform follows from the increase in the number of parameters to optimize, which are represented by matrices $T_0$, $T_1$, specific vector transformation $\f$ (Sections \ref{w0an} and \ref{det1}), and   $\w$-injection and $\h$-injection (Sections \ref{stat},  \ref{x78an} and \ref{speccases} ).

\section{Structure of the proposed PCA extension}\label{w0an}

The above advantages are achieved due to the special structure of the proposed transform  as follows.
   Let $\w\in  L^2(\Omega,\mathbb{R}^\ell)$ be a random vector  and $\f: L^2(\Omega,\mathbb{R}^n)\times L^2(\Omega,\mathbb{R}^\ell) \rightarrow L^2(\Omega,\mathbb{R}^\ell)$ be a transformation of $\y$ and $\w$ in a random vector  $\s\in  L^2(\Omega,\mathbb{R}^\ell)$, i.e.
$$
\s = \f(\y,\w).
$$
Reasons for using vector $\w$ and transformation $\f$ are detailed in Sections \ref{stat}, \ref{x78an} and \ref{speccases} below.

 We propose to determine the PCs and their reconstruction  $\widetilde{\x}$ by the  transform $\ttt $ given by
\begin{eqnarray}\label{sjk92}
\widetilde{\x}=\ttt (\y, \w) = T_0\y +  T_1 \f(\y, \w),
\end{eqnarray}
 where $T_0$ and $T_1$ are represented by $m\times n$ and $m\times \ell$ matrices, respectively, and
\begin{eqnarray}\label{xnm93}
\rank [T_0 \hspace*{1mm} T_1] \leq k,
 \end{eqnarray}
  where $k=\min\{m, n\}$. Here, three terms, $T_0$, $T_1$ and $\f$, are to be determined.
 Therefore, transform  $\ttt$ will be called the three-term $k$-rank transform.

A special version of the three-term $k$-rank transform is considered in Section \ref{xvb8an}  below.

\section{Statement of the problems}\label{stat}

 Below, we assume that $\x$, $\y$ and $\w$ are nonzero vectors.

{\em Problem 1}: Find matrices $T_0$ and $T_1$ that solve
\begin{eqnarray}\label{589mb}
\min_{[T_0 \hspace*{1mm} T_1]}\|\x - [ T_0\y +  T_1 \f(\y, \w)]\|^2_\Omega
 \end{eqnarray}
subject to constraint (\ref{xnm93}), and determine $\f$ that provides, for $\zz = [\y^T \s^T]^T$,
 \begin{eqnarray}\label{akl7}
E_{zz} =\left[ \begin{array}{cc}
                         E_{yy}  & \oo\\
                         \oo  & E_{ss}
                           \end{array} \right],
 \end{eqnarray}
 where $\oo$ denotes the zero matrix.
 The importance of the condition in (\ref{akl7}) is twofold. First, this allows us to facilitate computation associated with a determination of $T_0$ and $T_1$. Second, the condition in (\ref{akl7}) is used in the solution of the Problem 2 stated below.

 The transform obtained from the solution of Problem 1 (see Section \ref{det1}  that follows) is called the {\em optimal} three-term $k$-rank transform or the {\sf\em three-term PCA.}

 {\em Problem 2}: Show that the error associated with the  three-term PCA  is less than that of the PCA, GBT1 (see Section \ref{298an}),  GBT2 and GKLT (see Section \ref{xvb8an}). Further, show that the computational load associated with the tree-term PCA is less than that of the GBT2  (see Section \ref{58b29}).

%Further, in the practical setting or in the training stage, random vectors are replaced with their samples. Matrix $X$ is available and we wish to find its PCs ...  in this case, matrices $Y$ and $W$ are used to crate more parameters to optimize
\section{Differences from known techniques}\label{dif77}

The proposed three-term PCA differs from PCA in \cite{Jolliffe2002,Brillinger2001} in several instances. Unlike PCA
in \cite{Jolliffe2002,Brillinger2001},  the three-term PCA has the additional terms $T_1$, $\f$ and  $\w$-injection, which lead to the improvement in the associated  accuracy of determining PCs and  the consecutive reconstruction of PCs to the original vector. As distinct from the  PCA in \cite{Jolliffe2002,Brillinger2001}, the three-term PCA is always applicable to singular data since it is determined in terms of pseudo-inverse matrices.

Differences of the three-term PCA from the GBT2 are threefold. First, the three-term PCA contains  transformation $\f$ aimed to facilitate computation. Second, in the three-term PCA,
the procedure for determining principal components and their reconstruction to an estimate of $\x$ is different from that in the GBT2. Indeed, the three-term PCA can be written as
\begin{eqnarray}\label{sj382}
\ttt (\y, \w) = R_1 P_1\y +  R_2 P_2\s,
 \end{eqnarray}
where $R_1, R_2$, $P_1$ and $P_2$ are obtained in the form different from those in the GBT2 (see Theorem \ref{389nm} below).
 Third, in Section \ref{xvb8an} that follows, we show that a special transformation of vector $\s$ to a vector $\widetilde{\s}$ of a greater dimensionality allows us to achieve the better associated accuracy of $\x$ estimation.
Differences from the GKLT in (\ref{qmiw8}) are similar and  even stronger since the GKLT contains vector $\y^2$ (not $\vv$ as the GBT2) which cannot be changed.

The above differences imply the  improvement in the performance of the three-term PCA. This issue is detailed  in Sections \ref{x78an}  and \ref{x78kl} that follow.

%\section{Main results}
\section{Solution of Problem 1}\label{det}

\subsection{Preliminaries}\label{prel}

First, we recall some known  results that will be used in the solution of Problems 1 and 2.
%the derivation and justification of the proposed technique.

%Here, we provide some known results, which will be used in the derivation of solutions for Problems 1 and 2.
\begin{proposition} {\em \cite[Theorem 1.21, p. 44]{zhang2005schur}}\label{proposition5}
Let $M$ be a positive semi-definite matrix given in the block form
$
M=\left[
     \begin{array}{cc}
       A & B \\
       B^T & C \\
     \end{array}
   \right],
   $
where blocks $A$ and $C$ are square. Let $S=C-B^TA^\dagger B$ and
$$
N =\left[
     \begin{array}{cc}
       A^\dagger +A^\dagger BS^\dagger B^TA^\dagger  & -A^\dagger BS^\dagger  \\
       -S^\dagger B^TA^\dagger  & S^\dagger  \\
     \end{array}
   \right].
   $$
Then $N=M^\dagger$ if and only if $\mbox{rank}(M)=\mbox{rank}(A)+\mbox{rank}(C)$.
\end{proposition}

\begin{proposition}{\em \cite[p. 217]{Zhang2011}}\label{proposition6}
  If $M$ is positive definite, then the condition $\mbox{rank} (M)=\mbox{rank} (A)+\mbox{rank} (C)$ of Proposition \ref{proposition5} is always true and $N=M^{-1}$.
\end{proposition}

\begin{proposition}{\em \cite[Lemma 4.5.11]{harville2008matrix}}\label{proposition10}
For any matrices $A$ and $B$,
\begin{eqnarray}\label{zmnb91}
\mbox{\rm rank}\left[
                 \begin{array}{cc}
                   A & \oo \\
                   \oo & B \\
                 \end{array}
               \right]=\mbox{\rm rank}(A)+\mbox{\rm rank}(B).
\end{eqnarray}
\end{proposition}

\begin{proposition} {{\em (Weyl's inequality)}} {\em \cite[Corollary 4.3.15]{9780511810817}}\label{proposition8}
Let $A$ and $B$ be $m\times m$ symmetric matrices and let singular values  $\sigma_i(A)$, $\sigma_i(B)$ and $\sigma_i(A+B)$, for $i=1,\ldots, m$, be arranged in decreasing order. Then, for $i=1,\ldots, m$,
\begin{eqnarray}\label{7hb91}
\sigma_i(A)+\sigma_{m}(B)\leq \sigma_i(A+B)\leq \sigma_i(A)+\sigma_{1}(B).
\end{eqnarray}
%{\bf (should it be $\sigma_1(A+B)$?)}
%where $\sigma_i(A+B)$ and $\sigma_i(A)$ represent the $i$-th eigenvalue of $A+B$ and $A$, respectively, and $\sigma_{min}(B)$ and $\sigma_{max}(B)$  are the lowest and greatest eigenvalue of $B$, respectively.
\end{proposition}

\subsection{Determination of three-term PCA}\label{det1}

 %For $k<m$, $j<n$ and $\ell<\min(m,n)$, we denote $U_{M,k}:=[u_1 \;u_2\;\ldots u_k]$, $V_{M,j}=[v_1 \;v_2\;\ldots\vv_j]$ and $\Sigma_{M,\ell}=\diag(\sigma_1(M), \ldots,\sigma_\ell(M))$.
  Let $$\displaystyle P_{M,L}=\hspace*{-4mm}\sum_{k=1}^{\rank (M)}u_k u_k^T\in \rt^{m\times m}, \hspace*{3mm} \displaystyle  P_{M,R}=\hspace*{-4mm} \sum_{j=1}^{\rank (M)}v_j v_j^T\in \rt^{n\times n}$$
 be the orthogonal projections on the range of matrices $M$ and $M^T$ respectively, and let
\begin{eqnarray}\label{mkumk1}
 [M]_k= \sum_{i=1}^{k}\sigma_i(M)u_i v_i^T\in \rt^{m\times n}
 \end{eqnarray}
  for $k=1,\ldots,\rank (M)$, be the truncated SVD of $M$.
 For $k>\rank (M)$, we define $[M]_k=M\;(=M_{\rank (M)})$. For $1\le k<\rank (M)$, the matrix $M_k$ is uniquely defined
 if and only if $\sigma_k(M)>\sigma_{k+1}(M)$.

Further, $M^{1/2}$ denotes a matrix square root for matrix $M$.
For the covariance matrix $E_{x x}$, we denote $E_{x  x}^{1/2 \dag}:=(E_{ x x}^{1/2})^ {\dag}$. Matrix $E_{ x x}^{1/2 \dag}$ is unique since $E_{x x}$ is positive semidefinite.
The Frobenius matrix norm is denoted by $\|\cdot\|$.

Let us denote  $G_{xy}=E_{xy}E_{yy}^\dag$ and  $G_z = E_{x z}E_{zz}^{\dag}E_{z x}$. Similar to $U_{G_q,k}$ in (\ref{ac11}), $U_{G_z,k}$ denotes the matrix  formed by the first $k$ columns of $U_{G_z}$. Recall that, as before in (\ref{akl7}), $\zz=[\y^T \s^T]^T$.
% write $G_z$ and its SVD similar to $G_y$ and its SVD  in Sections \ref{gbt} and \ref{svd} above.

\begin{theorem}\label{389nm}
Let  transformation $\f$ in (\ref{sjk92}) be determined by
\begin{eqnarray}\label{smi91}
\f(\y, \w) = \s = \w - G_{wy} \y.
 \end{eqnarray}
 Then (\ref{akl7}) is true, and $T_0$ and $T_1$ that solve the problem in (\ref{589mb}), (\ref{xnm93}) are such that
 \begin{eqnarray}\label{z,.23}
[T_0\hspace*{1.5mm} T_1] = U_{G_z,k}U_{G_z,k}^T [G_{xy}\hspace*{1.5mm} G_{xs}](I + N)
%& =& [ U_{G_z,k}U_{G_z,k}^TG_{xy}\hspace*{1.5mm}  U_{G_z,k}U_{G_z,k}^TG_{xs}](I + N)
 \end{eqnarray}
where
 $N= M(I - P_{E_{zz}^{1/2},L})$ and matrix $M$ is arbitrary\footnote{In other words, the solution is not unique.}. The unique minimum norm solution of problem (\ref{589mb}), (\ref{xnm93}) is given by
 \begin{eqnarray}\label{294b}
T_0 = U_{G_z,k}U_{G_z,k}^T G_{xy}\qa T_1 = U_{G_z,k}U_{G_z,k}^T  G_{xs},
 \end{eqnarray}
 where
 \begin{eqnarray}\label{2mox5}
G_z = G_y + G_s.
 \end{eqnarray}
 \end{theorem}
%%%%%%%%%%%%%%%%%%

\begin{proof}
For vector $\s$ defined by (\ref{smi91}),
$$
E_{ys} = E[\y (\w - E_{wy}E_{yy}^\dag \y)^T]=E_{yw} - E_{yy}E_{yy}^\dag E_{yw}=\oo
$$
because by Corollary 1 in {\em\cite{924074}}, $E_{yw} = E_{yy}E_{yy}^\dag E_{yw}$. Then (\ref{akl7}) follows. Further, since $\|\x\|^2_\Omega = \mbox{\em tr} \hspace*{.5mm}E[\x \x^T]$ (see {\em\cite[pp. 166-167]{torbook2007}}) then, for $T=[T_0\hspace*{.5mm} T_1]$,
 \begin{eqnarray}\label{wnmui}
\|\x - [ T_0\y +  T_1 \f(\y, \w)]\|^2_\Omega &=& \|\x - T\zz\|^2_\Omega\nonumber\\
&=&  \mbox{\em tr}\hspace*{.5mm} E\{(\x - T\zz) (\x - T\zz)^T\}\nonumber\\
&=& \|E_{xx}^{1/2}\|^{2} -  \|E_{xz}{E_{zz}^{1/2}}^\dag\|^{2} \nonumber\\
& &\hspace*{17mm} + \|E_{xz}{E_{zz}^{1/2}}^\dag - T{E_{zz}^{1/2}}\|^{2}.
 \end{eqnarray}
Therefore, the problem in (\ref{589mb})--(\ref{xnm93}) is reduced to
\begin{eqnarray}\label{}
\min_{T:  \rank T \leq k} \|E_{xz}{E_{zz}^{1/2}}^\dag - T{E_{zz}^{1/2}}\|^{2}.
 \end{eqnarray}
 Its solution is given in \cite{tor5277} by
   \begin{eqnarray}\label{xm02}
T=[T_0 \hspace*{1mm} T_1] = [E_{xz}{E_{zz}^\dag}^{1/2}]_k{E_{zz}^\dag}^{1/2}(I + N).
  \end{eqnarray}
Let us write $U_Q\Sigma_Q V_Q^T = Q$  for the SVD of $Q=E_{xz}{E_{zz}^\dag}^{1/2}$.
Then by \cite{tor843},
 \begin{eqnarray}\label{amn82}
 [E_{xz}{E_{zz}^\dag}^{1/2}]_k = U_{Q, k} U_{Q, k}^T E_{xz}{E_{zz}^\dag}^{1/2}.
 \end{eqnarray}
 Since ${G_z}=Q Q^T$ then $U_{G_z}=U_Q$ and $U_{{G_z}, k} = U_{Q, k}$. Therefore, (\ref{xm02}) and (\ref{amn82}) imply
\begin{eqnarray}\label{9cn02}
T=[T_0 \hspace*{1mm} T_1] = U_{{G_z}, k}U_{{G_z}, k}^T G_{xz}(I + N).
  \end{eqnarray}
 Here, on the basis of (\ref{akl7}),
\begin{eqnarray}\label{am198n}
 G_{xz} = [E_{xy}\hspace*{1.5mm}  E_{xs}]\left[ \begin{array}{cc}
                         E_{yy}^\dag  & \oo\\
                         \oo  & E_{ss}^\dag
                           \end{array} \right] =  [G_{xy}\hspace*{1.5mm} G_{xs}]
\end{eqnarray}
and
 \begin{eqnarray}\label{348vn}
 &&\hspace*{-30mm} {G_z}=[E_{xy}\hspace*{1.5mm}  E_{xs}]\left[ \begin{array}{cc}
                         E_{yy}^\dag  & \oo\\
                         \oo  & E_{ss}^\dag
                           \end{array} \right] \left[ \begin{array}{c}E_{yx}\\  E_{sx}\end{array} \right] \nonumber\\
 &&\hspace*{-20mm} = E_{xy}E_{yy}^\dag E_{yx} +  E_{xs}E_{ss}^\dag E_{sx}= G_{y}  + G_{s}.
\end{eqnarray}
 Then (\ref{z,.23}), (\ref{294b}) and (\ref{2mox5}) follow.
$\hfill\blacksquare$   \end{proof}

Thus, the three-term PCA is represented by (\ref{sjk92}), (\ref{smi91}), (\ref{z,.23}) and (\ref{294b}).

\section{Analysis of the error associated with three-term PCA}

Let us denote the error associated with the three-term PCA by
\begin{eqnarray}\label{s89mb}
\varepsilon_{m,n,\ell} (T_0, T_1) = \min_{\substack{[T_0 \hspace*{1mm} T_1]: \\ \rank [T_0 \hspace*{1mm} T_1] \leq k}}\|\x - [ T_0\y +  T_1 \f(\y, \w)]\|^2_\Omega.
 \end{eqnarray}
The following theorem establishes  a priori determination of $\varepsilon_{m,n,\ell} (T_0, T_1)$.

%%%%%%%%%%%%%%%%%%%%%%%%%%%%%%%%%%%%%%%%%

\begin{theorem}\label{377nm}
Let $\f$, $T_0$ and $T_1$  be determined by Theorem \ref{389nm}. Then
 \begin{eqnarray}\label{215smb}
\varepsilon_{m,n,\ell} (T_0, T_1)  = \|E_{xx}^{1/2}\|^{2} - \sum_{i=1}^{k} \sigma_i ({G_z}).
 \end{eqnarray}
 \end{theorem}
%%%%%%%%%%%%%%%%%%

\begin{proof}
It follows from (\ref{wnmui}) and (\ref{xm02})  that
 \begin{eqnarray}\label{q28mb}
&&\hspace*{-10mm}\varepsilon_{m,n,\ell} (T_0, T_1) \nonumber\\
&=& \|E_{xx}^{1/2}\|^{2} -  \|E_{xz}{E_{zz}^{1/2}}^\dag\|^{2} + \|E_{xz}{E_{zz}^{1/2}}^\dag - [E_{xz}{E_{zz}^\dag}^{1/2}]_k{E_{zz}^\dag}^{1/2}(I + N){E_{zz}^{1/2}}\|^{2}\nonumber\\
&=& \|E_{xx}^{1/2}\|^{2} -  \|E_{xz}{E_{zz}^{1/2}}^\dag\|^{2} + \|E_{xz}{E_{zz}^{1/2}}^\dag - [E_{xz}{E_{zz}^\dag}^{1/2}]_k
\|^{2}.
%= \mbox{\em tr} (E_{xx}) - \sum_{i=1}^{k} \sigma_i (G).
  \end{eqnarray}
The latter is true because by Lemma 42 in \cite[p. 311]{torbook2007},
$$
[E_{xz}{E_{zz}^\dag}^{1/2}]_k = [E_{xz}{E_{zz}^\dag}^{1/2}]_k{E_{zz}^\dag}^{1/2}{E_{zz}^{1/2}}.
$$
Then
\begin{eqnarray}\label{qr56b}
\varepsilon_{m,n,\ell} (T_0, T_1)& = & \|E_{xx}^{1/2}\|^{2} - \sum_{i=1}^{m} \sigma_i({G_z}) + \sum_{i=k+1}^{m} \sigma_i({G_z})\nonumber\\
                                 & = & \|E_{xx}^{1/2}\|^{2} - \sum_{i=1}^{k} \sigma_i({G_z}).
  \end{eqnarray}
  Thus, (\ref{215smb}) is true.\hfill$\blacksquare$
\end{proof}

\section{Advantages of three-term PCA}\label{x78an}

Here and in Section \ref{x78kl} below, we justify in detail the  advantages of the three-term PCA that have been highlighted in Section \ref{wn9w}.

\subsection{Solution of Problem 2. Improvement in the associated error compared to PCA and GBT1}\label{298an}

We wish to show that the error associated with the three-term PCA, ${\varepsilon_{m,n,\ell} (T_0, T_1)}$, is  less  than that of PCA and the GBT1 \cite{tor843}. A similar statement has been provided in Corollary 3 in \cite{924074} under the condition which is difficult to verify. In Theorems \ref{xnm91n} and \ref{mak92n} below, we show that the condition can be omitted. Let us denote the error associated with the GBT1 by
%Recall that the error associated with the  GBT  $F$ is given by
\begin{eqnarray}\label{3290b}
\varepsilon_{m,n} (B_0) = \min_{\substack{B_0\in{\mathbb R}^{m\times n}: \\ \rank (B_0) \leq k}}\|\x - B_0\y\|^2_\Omega.
  \end{eqnarray}
Matrix $B_0=R_1P_1$ that solves the RHS in (\ref{3290b}) follows from (\ref{ac11}) if  $R_1 P_2\vv=\mathbf 0$ (see Section \ref{gbt1}).

%{\bf (Advantage over GBT2? In Theorem below, there is no condition... for $\leq$ ...)}
%{\bf (Are the errors for GKLT and GBT1 the same???? -  They are!)}
%%%%%%%%%%%%%%%%%%%%%%%%%%%%%%%%%%%%%%%%%

\begin{theorem}\label{xnm91n}
For any non-zero random vectors $\x, \y$ and $\w$,
\begin{eqnarray}\label{al9b}
\varepsilon_{m,n,\ell} (T_0, T_1) \leq \varepsilon_{m,n} (B_0).
  \end{eqnarray}
  If $G_s=E_{xs}E_{ss}^\dag E_{sx}$ is positive definite then
  \begin{eqnarray}\label{al9xc}
\varepsilon_{m,n,\ell} (T_0, T_1) < \varepsilon_{m,n} (B_0).
  \end{eqnarray}
\end{theorem}
%%%%%%%%%%%%%%%%%%

\begin{proof}
It is known  {\em \cite{tor843}} that
\begin{eqnarray}\label{xm819}
\varepsilon_{m,n} (B_0) = \|E_{xx}^{1/2}\|^{2} - \sum_{i=1}^{k} \sigma_i ({G_y}).
  \end{eqnarray}
Consider ${G_z} = G_y + ({G_z} - {G_y}).$ Clearly, ${G_z} - {G_y}$ is a symmetric matrix. Then on the basis of (\ref{7hb91}) in the above Proposition \ref{proposition8},
\begin{eqnarray}\label{20mna}
\sigma_i(G_y)+\sigma_{m}({G_z} - {G_y})\leq \sigma_i(G_z),
\end{eqnarray}
where
$$
{G_z} - {G_y} = G_s=MM^T
$$
and $M =  E_{xs}{E_{ss}^\dag}^{1/2}$. Thus, by Theorem 7.3 in \cite{Zhang2011}, ${G_z} - {G_y}$
 is a positive semi-definite matrix and then all its eigenvalues are nonnegative \cite[p. 167]{golub2013matrix}, i.e., $\sigma_i({G_z} - {G_y})\geq 0.$
Therefore, (\ref{20mna}) implies
$
\sigma_i ({G_y})\leq \sigma_i ({G_z}),
$
for all $i=1,\ldots,n$, and then
\begin{eqnarray}\label{akl81}
\sum_{i=1}^k\sigma_i ({G_y})\leq \sum_{i=1}^k\sigma_i ({G_z}).
\end{eqnarray}
As a result, (\ref{al9b}) follows from (\ref{215smb}), (\ref{xm819}) and (\ref{akl81}). In particular, if $G_s$ is positive definite then  $\sigma_i({G_z} - {G_y}) > 0$, for $i=1,\ldots,n$ and therefore, (\ref{al9xc}) is true.\hfill$\blacksquare$
\end{proof}

In the following Theorem \ref{mak92n}, we refine the result obtained in the above Theorem \ref{xnm91n}. %Let $r_s=\rank G_s$.

\begin{theorem}\label{mak92n}
Let,  as before, $k= \min \{m, n\}.$ There exists  $\gamma\in [\sigma_m(G_s), \sigma_{1}(G_s)]$ such that
\begin{eqnarray}\label{al917}
\varepsilon_{m,n,\ell} (T_0, T_1) = \varepsilon_{m,n} (B_0) - k\gamma,
  \end{eqnarray}
  i.e., the error associated with the three-term PCA is less than that of the GBT1  by $k\gamma$.
\end{theorem}

\begin{proof}
The Weyl's inequality in (\ref{7hb91}) and the equality in (\ref{348vn})  imply, for $i= 1,\ldots, m$,
\begin{eqnarray*}\label{7dmnf1}
\sigma_i(G_y)+\sigma_{m}(G_s)\leq \sigma_i(G_z)\leq \sigma_i(G_y)+\sigma_{1}(G_s)
\end{eqnarray*}
which, in turn, implies
\begin{eqnarray*}
\sigma_{m}(G_s)\leq \sigma_i(G_z) - \sigma_i(G_y)\leq \sigma_{1}(G_s)
\end{eqnarray*}
and
\begin{eqnarray*}
\sum_{i=1}^k\sigma_{m}(G_s)\leq \sum_{i=1}^k[\sigma_i(G_z) - \sigma_i(G_y)]\leq \sum_{i=1}^k\sigma_{1}(G_s).
\end{eqnarray*}
Therefore,
\begin{eqnarray*}
k\sigma_{m}(G_s)\leq \sum_{i=1}^k\sigma_i(G_z) - \sum_{i=1}^k\sigma_i(G_y)\leq k\sigma_{1}(G_s)
\end{eqnarray*}
and
\begin{eqnarray*}
k\sigma_{m}(G_s)\leq \left(\|E_{xx}^{1/2}\|^{2} - \sum_{i=1}^k\sigma_i(G_y)\right) - \left(\|E_{xx}^{1/2}\|^{2} - \sum_{i=1}^k\sigma_i(G_z)\right)\leq k\sigma_{1}(G_s).
\end{eqnarray*}
Thus
\begin{eqnarray*}
k\sigma_{m}(G_s)\leq \varepsilon_{m,n} (B_0) - \varepsilon_{m,n,\ell} (T_0, T_1) \leq k\sigma_{1}(G_s)
\end{eqnarray*}
and
\begin{eqnarray*}
 \frac{\varepsilon_{m,n} (B_0) - \varepsilon_{m,n,\ell} (T_0, T_1)}{k} \in [\sigma_{m}(G_s), \sigma_{1}(G_s)],
\end{eqnarray*}
and then (\ref{al917}) follows.
$\hfill\blacksquare$   \end{proof}

\begin{remark}
If  $G_s$ is a full rank matrix then $\sigma_{m}(G_s)\neq 0$ and therefore, $\gamma\neq 0$, i.e. in this case,  (\ref{al917}) implies that $\varepsilon_{m,n,\ell} (T_0, T_1)$ is always less than $\varepsilon_{m,n} (B_0)$.
If $\rank (G_s) = r_s$ where $r_s <m$ then  $\sigma_{m}(G_s) = 0$ and $\gamma\in [0, \sigma_{1}(G_s)]$, i.e. in this case, $\gamma$ might be equal to $0$.
\end{remark}

\begin{remark}
Recall that PCA is a particular case of the GBT1 (see Section \ref{gbt1}). Therefore, in Theorems \ref{xnm91n} and \ref{mak92n}, $\varepsilon_{m,n} (B_0)$ can be treated as the error associated with PCA, under the restriction that matrix $E_{yy}$ is non-singular.
\end{remark}

%%%%%%%%%%%%%%%%%%%%%%%%%%%%%%%%%%%%%%%%%%%%%%%%%%%%%%%%%%%%%%%%%%%%%%%%%%%%%%%%%%%%%%%%%%%%%%
\subsection{ Decrease in the error associated with the three-term PCA with
 the  increase in the injection dimension}\label{xvb8an}
  % of `auxiliary' vector}\label{xvb8an}
%%%%%%%%%%%%%%%%%%%%%%%%%%%%%%%%%%%%%%%%%%%%%%%%%%%%%%%%%%%%%%%%%%%%%%%%%%%%%%%%%%%%%%%%%%%%%%%%%%

 In Theorem \ref{w791n} that follows we show that the error $\varepsilon_{m,n,\ell} (T_0, T_1)$ associated with the three-term PCA (represented by  (\ref{sjk92}), (\ref{smi91})--(\ref{294b})) can be decreased  if vector $\s$ is extended to a new vector $\widetilde\s$ of a  dimension which is larger than that of vector $\s$. The vector  $\widetilde\s$ is constructed as $\widetilde\s = [\s^T  \gf^T]^T$ where $\gf = \h - G_{h z}\zz\in  L^2(\Omega,\mathbb{R}^\eta)$, $G_{h z}=E_{h z}E_{z z}^\dag$ and $\h\in  L^2(\Omega,\mathbb{R}^\eta)$ is arbitrary. As we mentioned before, similar to $\w$-injection, vector $\h$ is called the $\h$-injection. As before, $\s$ is defined by (\ref{smi91}) and  $\zz=[\y^T \s^T]^T$.
Thus, $\widetilde\s \in  L^2(\Omega,\mathbb{R}^{(\ell+\eta)})$ while  $\s \in  L^2(\Omega,\mathbb{R}^\ell)$, i.e. the dimension of $\widetilde\s$ is larger than that of $\s$ by $\eta$ entries.
In terms of  $\widetilde\s$, the three-term PCA is represented as
\begin{eqnarray}\label{sjk49}
\sss (\y, \w,\h) = S_0\y +  S_1\widetilde\s,
 \end{eqnarray}
 where similar to $T_0$ and $T_1$ in (\ref{294b}), and for $\widetilde{\zz}=[\y^T \hspace*{1mm} \widetilde{\s}^T]^T$,   matrices $S_0$ and $S_1$ are given by
 \begin{eqnarray}\label{wn202}
S_0 = U_{G_{\widetilde{z}},k}U_{G_{\widetilde{z}},k}^T G_{xy}\qa S_1 = U_{G_{\widetilde{z}},k}U_{G_{\widetilde{z}},k}^T  G_{x {\widetilde{s}}}.
 \end{eqnarray}
Here, $G_{\widetilde{z}} = G_y + G_{\widetilde{s}}$ and
$G_{\widetilde{s}} = E_{x {\widetilde{s}}}E_{{\tilde{s}}{\tilde{s}}}^{\dag}E_{{\widetilde{s}} x}$.
The associated error is denoted by
  \begin{eqnarray}\label{cn920}
\varepsilon_{m,n,\ell+\eta} (S_0, S_1) = \min_{\substack{[S_0 \hspace*{1mm} S_1]\in{\mathbb R}^{m\times (n+\ell+\eta)}: \\ \rank [S_0 \hspace*{1mm} S_1] \leq k}}\|\x - [ S_0\y +  S_1\widetilde\s]\|^2_\Omega.
\end{eqnarray}

\begin{theorem}\label{w791n}
For any non-zero random vectors $\x, \y,\w$ and $\h$,
\begin{eqnarray}\label{cnm201}
\varepsilon_{m,n,\ell+\eta} (S_0, S_1) \leq \varepsilon_{m,n,\ell} (T_0, T_1).
\end{eqnarray}
 If $G_s=E_{xs}E_{ss}^\dag E_{sx}$ is positive definite then
 \begin{eqnarray}\label{q87201}
\varepsilon_{m,n,\ell+\eta} (S_0, S_1) < \varepsilon_{m,n,\ell} (T_0, T_1).
\end{eqnarray}
\end{theorem}
%%%%%%%%%%%%%%%%%%

\begin{proof}
Let us represent $S_1$ in terms of two blocks, $S_{11}$ and $S_{12}$, i.e., $S_1 = [S_{11}\hspace*{1mm} S_{12}],$ and also  write ${\widehat S} = [S_{0}\hspace*{1mm} S_{11}]$. Then
\begin{eqnarray}\label{190cn3}
& &\hspace*{-30mm}  \min_{\substack{[S_0 \hspace*{1mm} S_1]\in{\mathbb R}^{m\times (n+\ell+\eta)}: \\ \rank [S_0 \hspace*{1mm} S_1] \leq k}}\left\|\x - [ S_0\y +  S_1\widetilde\s]\right\|^2_\Omega\nonumber\\
& = & \min_{\substack{[S_0 \hspace*{1mm} S_1]\in{\mathbb R}^{m\times (n+\ell+\eta)}: \\ \rank [S_0 \hspace*{1mm} S_1] \leq k}}\left\|\x - \left[ S_0\y +  S_1\left[
                 \begin{array}{c}
                   \s \\
                   \gf \\
                 \end{array}
               \right]\right]\right\|^2_\Omega \nonumber\\
            & = & \min_{\substack{[S_0 \hspace*{1mm} S_1]\in{\mathbb R}^{m\times (n+\ell+\eta)}: \\ \rank [S_0 \hspace*{1mm} S_1] \leq k}}\left\|\x - [ S_0\y +  S_{11}\s + S_{12}\gf ]\right\|^2_\Omega  \nonumber \\
               & = &\min_{\substack{[S_0 \hspace*{1mm} S_1]\in{\mathbb R}^{m\times (n+\ell+\eta)}: \\ \rank [S_0 \hspace*{1mm} S_1] \leq k}} \left\|\x - [{\widehat S}\zz + S_{12}\gf ]\right\|^2_\Omega.
\end{eqnarray}
Here,
 $[S_0 \hspace*{1mm} S_1] = [S_0 \hspace*{1mm} S_{11}\hspace*{1mm} S_{12}] = [{\widehat S}\hspace*{1mm} S_{12}]$. Therefore,
\begin{eqnarray}\label{3333m}
& &\hspace*{-37mm} \min_{\substack{[S_0 \hspace*{1mm} S_1]\in{\mathbb R}^{m\times (n+\ell+\eta)}: \\ \rank [S_0 \hspace*{1mm} S_1] \leq k}} \left\|\x - [{\widehat S}\zz + S_{12}\gf ]\right\|^2_\Omega\nonumber\\
& = & \min_{\substack{[{\widehat S} \hspace*{1mm} S_{12}]\in{\mathbb R}^{m\times (n+\ell+\eta)}: \\ \rank [{\widehat S} \hspace*{1mm} S_{12}] \leq k}}\|\x - [{\widehat S}\zz + S_{12}\gf ]\|^2_\Omega.
\end{eqnarray}
In (\ref{s89mb}), let us write $\varepsilon_{m,n,\ell} (T_0, T_1)$  in terms of $\s$,
  \begin{eqnarray}\label{qg59mb}
\varepsilon_{m,n,\ell} (T_0, T_1) = \min_{\substack{[T_0 \hspace*{1mm} T_1]\in{\mathbb R}^{m\times (n+\ell)}: \\ \rank [T_0 \hspace*{1mm} T_1] \leq k}}\|\x - [ T_0\y +  T_1 \s]\|^2_\Omega.
 \end{eqnarray}
Then by (\ref{al9b}) in Theorem \ref{xnm91n},
\begin{eqnarray}\label{po02m}
& & \hspace*{-35mm}\min_{\substack{[{\widehat S} \hspace*{1mm} S_{12}]\in{\mathbb R}^{m\times (n+\ell+\eta)}: \\ \rank [{\widehat S} \hspace*{1mm} S_{12}] \leq k}}\|\x - [{\widehat S}\zz + S_{12}\gf ]\|^2_\Omega \nonumber\\
&\leq &\min_{\substack{T\in{\mathbb R}^{m\times (n+\ell)}: \\ \rank (T) \leq k}}\|\x - T\zz\|^2_\Omega\nonumber\\
& = &\min_{\substack{[T_1 \hspace*{1mm} T_2]\in{\mathbb R}^{m\times (n+\ell)}: \\ \rank [T_1 \hspace*{1mm} T_2] \leq k}}\left\|\x - [T_0 \hspace*{1mm} T_1]\left[
                 \begin{array}{c}
                   \y \\
                   \s \\
                 \end{array}
               \right]\right\|^2_\Omega\nonumber\\
& = &\min_{\substack{[T_0 \hspace*{1mm} T_1]\in{\mathbb R}^{m\times (n+\ell)}: \\ \rank [T_0 \hspace*{1mm} T_1]\leq k}} \|\x - [T_0\y +T_1\s]\|^2_\Omega,
\end{eqnarray}
where $T = [T_0 \hspace*{1mm} T_1]$, $T_0\in{\mathbb R}^{m\times n}$ and $T_1\in{\mathbb R}^{m\times \ell}$.
Then (\ref{cnm201})  follows from (\ref{3333m}) and  (\ref{po02m}),  and  (\ref{al9xc}) implies (\ref{q87201}).
$\hfill\blacksquare$   \end{proof}

\begin{remark}\label{hjk29}
An intuitive explanation of the statement of Theorem \ref{w791n} is that the increase in the  dimension of vector $\widetilde\s$ implies the increase in the dimension of matrix $S_1$ in (\ref{sjk49}) so that $S_1\in\rt^{m\times (n+\ell+\eta)}$ while in (\ref{sjk92}),  (\ref{smi91})--(\ref{294b}), $T_1\in\rt^{m\times (n+\ell)}$. Therefore, the optimal matrix $S_1$ has $m\times \eta$ entries more  than $T_1$ to further  minimize the associated error. As a result, the three-term PCA in form  (\ref{sjk49}), for the same number of principal components $k$, provides the more accurate reconstruction of $\x$ than the three-term PCA in form (\ref{sjk92}).
\end{remark}

%%%%%%%%%%%%%%%%%%%%%%%%%%%%%%%%%%%%%%%%%%%%%%%%%%%%%%%%%%%%%%%%%%%%%%%%%%%%%%%%%%%%%%%%%%%%%%
\subsection{Decrease in the error associated with the three-term PCA  compared with that of the GBT2 and GKLT}\label{1n8an}
%%%%%%%%%%%%%%%%%%%%%%%%%%%%%%%%%%%%%%%%%%%%%%%%%%%%%%%%%%%%%%%%%%%%%%%%%%%%%%%%%%%%%%%%%%%%%%%%%%

In Theorem \ref{wd87n} below, we show that  the three-term PCA in (\ref{sjk49})-(\ref{wn202}) provides the associated accuracy which is better than that of the GBT2 in (\ref{bnm11})-(\ref{ac11}).
Let us denote the error associated with the GBT2  by
\begin{eqnarray}\label{ii9m}
\varepsilon_{m,n} (B_0, B_1) = \min_{\substack{[B_0 \hspace*{1mm} B_1]\in{\mathbb R}^{m\times (2n)}: \\ \rank [B_0 \hspace*{1mm} B_1]\leq k}} \|\x - [B_0\y +B_1\vv]\|^2_\Omega,
\end{eqnarray}
where $B_0=R_1  P_1$ and $B_1=R_1  P_2$ are determined by (\ref{ac11}).
In fact, the result below is a version of Theorem \ref{w791n} as follows.

\begin{theorem}\label{wd87n}
Let $\s = \vv$ where $\ell = n$, and random vector $\vv$ is the same as in (\ref{bnm11})-(\ref{ac11}).
Then for any non-zero random vectors $\x, \y$ and $\h$,
\begin{eqnarray}\label{cnm101}
\varepsilon_{m,n,n+\eta} (S_0, S_1) \leq \varepsilon_{m,n} (B_0, B_1).
  \end{eqnarray}
 If $G_s=E_{xs}E_{ss}^\dag E_{sx}$ is positive definite then
 \begin{eqnarray}\label{n29h01}
\varepsilon_{m,n,n+\eta} (S_0, S_1) < \varepsilon_{m,n} (B_0, B_1).
  \end{eqnarray}
\end{theorem}
%%%%%%%%%%%%%%%%%%

\begin{proof}

We observe that  Theorem \ref{w791n} is true for any form of $\s$. In particular, it is true for $\s = \vv$ where $\ell = n$ and random vector $\vv$ is the same as in (\ref{bnm11})-(\ref{ac11}). Then (\ref{cnm101}) and (\ref{n29h01}) follow from (\ref{cnm201}) and (\ref{q87201}), respectively. $\hfill\blacksquare$
\end{proof}
  \begin{remark}\label{op28m}
Theorem \ref{wd87n} is also valid for $\s  = \y^2$ with $\ell =n$. Thus, in this case, the three-term PCA in (\ref{sjk49})-(\ref{wn202}) provides the associated accuracy which is better than that of the GKLT in (\ref{qmiw8}).
  \end{remark}

\section{Special cases of  three-term PCA}\label{speccases}

%\subsubsection{Choice of initial $V^{(0)}$ in the optimization procedure..} \label{09omk2}

%In Theorems  \ref{389nm} and \ref{377nm},
In both forms of the three-term PCA represented by (\ref{sjk92}) and (\ref{sjk49}), vectors $\s$ and $\widetilde \s$ are constructed from auxiliary random vectors called  $\w$-injection and $\h$-injection, respectively, which are assumed to be arbitrary. At the same time, a natural desire is to understand if there are approaches for choosing $\w$ and $\h$ which may (or may not) improve the three-term PCA performance. Here, such approaches are considered.

\subsection{Case 1.  Choice of $\w$ on the basis of  Weierstrass theorem}\label{}

  For the three-term PCA represented by (\ref{sjk92}), a seemingly reasonable choice of vector $\w$ is $\w = \y^2$ where $\y^2$ is defined by the Hadamard product, $\y^2 = \y\circ \y$, i.e. by $\y^2(\omega) = [\y^2_1(\omega), \ldots, \y^2_n(\omega)]^T$, for all $\omega\in \Omega$. This is because if  $\ttt(\y, \w)$ in (\ref{sjk92}) is written as $\ttt(\y, \y^2) = T_0\y + T_1 \f(\y,\y^2)$, then $\ttt(\y, \y^2)$ can be interpreted as a polynomial of the second degree. If  $T_0$ and  $T_1$  are defined by Theorem \ref{389nm}  then $\ttt(\y, \y^2)$  can be considered as an approximation to  an `idealistic' transform $\p$ such that $\x=\p(\y)$. On the basis of the Stone-Weierstrass theorem \cite{Kreyszig1978,Timofte2005}, $\ttt(\y, \y^2)$ should seemingly provide a better associated approximation accuracy of  $\p(\y)$  than that of the first degree polynomial $\ttt(\y) = T_0\y$. Nevertheless, the constraint  of the reduced ranks  implies a deterioration  of
$\p(\y)$ approximation by $\ttt(\y, \y^2) = T_0\y + T_1\y^2$. That constraint is not a condition of the Stone-Weierstrass theorem. To the best of our knowledge, an extension  of the Stone-Weierstrass theorem to a best  {\em rank constrained} estimation of $\x$ is still not justified.
 %Further, if $\vv$ is similar to $\x$, then .....
Further, if for example, $\y = \x +\bxi$ where $\bxi$ is a random noise, then $\y^2 = \x^2 +2\x\circ\bxi + \bxi^2$, i.e., $\y^2$ becomes even more corrupted than $\y$. Another inconvenience is that a knowledge or evaluation of matrices $E_{x y^2}$ and  $E_{y^2 y^2}$ is difficult. For the above reasons, the choice of $\w$ in the form $\w = \y^2$ is not preferable.

\subsection{Case 2.  Choice of $\w$ as an optimal estimate of $\x$}\label{}

Another  seemingly reasonable choice of  $\w$-injection in (\ref{sjk92}) is $\w=A\y$ where $A=E_{xy}E_{yy}^\dag$, i.e. $\w=E_{xy}E_{yy}^\dag \y$ is the optimal minimal-norm linear estimation of $\x$ \cite{torbook2007}. Nevertheless, in this case, $\s = A(\y - E_{yy}E_{yy}^\dag \y)$ and then
$$
 E_{xs}  = (E_{xy} -E_{xy}E_{yy}^\dag E_{yy})A^T = \oo
$$
since $E_{xy} =E_{xy}E_{yy}^\dag E_{yy}$ \cite[p. 168]{torbook2007}. The latter   implies $E_{xs} E_{ss}^\dag =  \oo$. As a result, by (\ref{294b}), $T_1=\oo$  and then in (\ref{sjk92}), $ \ttt(\y,\w) = T_0\y$. In other words, this choice of $\w$ is unreasonable since then the three-term PCA in (\ref{sjk92}), (\ref{294b}) is reduced to the GBT1 in \cite{tor843}.

\subsection{ Case 3. Choice of $\h$: `worse is better'}\label{vnb01m}

% Note that in Corollary \ref{gjk82}, term $\displaystyle\sum_{\ell =1}^{r_2} \sigma^2_\ell (S_z)$  depends on $\vv$ and its dimension $q$.
% An observation of  (\ref{vbn92})  leads to a quite surprising conclusion:  for {{\em any}} random `auxiliary' signal $\vv$, the error associated with the proposed optimal two-term transform $\f(\y,\vv)$ is {{\em always }} less than that of the optimal one-term transform $\widehat{F}\y$ in (\ref{4sjk92}). Moreover, as

It has been shown in Theorem \ref{w791n} that  the error associated with the three-term PCA represented by  (\ref{sjk49})  decreases if vector $\s\in  L^2(\Omega,\mathbb{R}^{\ell})$ is replaced with a new vector $\widetilde{\s}\in  L^2(\Omega,\mathbb{R}^{(\ell+\eta)})$ of a larger dimension. Recall, in (\ref{sjk49}), $\widetilde{\s}$ is formed  from an arbitrary $\h\in  L^2(\Omega,\mathbb{R}^{\eta})$.  In particular, for $\eta = 0$, (\ref{cnm201}) implies $\varepsilon_{m,n,\ell+\eta} (S_0, S_1) = \varepsilon_{m,n,\ell} (T_0, T_1)$. Thus, the increase in $\eta$ implies the decrease in the error associated with the three-term PCA represented by  (\ref{sjk49}).
 Thus, a reasonable (and quite surprising)  choice of $\h$ is as follows: $\h$ is random and $\eta$ is large.  This is similar to the  concept `worse is better' \cite{Gabriel}, i.e., a random $\h$  with large dimension $\eta$ (`worse') is a preferable option (`better') in terms of practicality and usability over, for example, the choice of $\w$ considered, for instance, in Case 1. In Case 1, the dimension  of $\w=\y^2$ is $n$ and cannot be changed while dimension $\ell$ can vary and, in particular, can be increased. Another advantage over Case 1 is that there is no need to evaluate matrices $E_{x y^2}$ and  $E_{y^2 y^2}$ as required in Case 1. %In some cases, a knowledge or evaluation of matrices $E_{x y^2}$ and  $E_{y^2 y^2}$ is difficult.

%in comparison with Case 1, matrices  $E_{x v}$ and  $E_{vv}$ are  given and there is no need to evaluate $E_{x y^2}$ and  $E_{y^2 y^2}$.

In Example \ref{m2b9} that follows,    we  numerically illustrate Theorem \ref{w791n} and Case 3.

%%%%%%%%%%%%%%%%%%%%%%%%%%%%%%%%%%%%%%%%%%%%%%%%%%%%%%%%%%%%%%%%%%%%
\begin{example}\label{m2b9}
Let ${\bf{x}}=[\x_1,\ldots, \x_m]^T\in L^2(\Omega,\mathbb{R}^{m})$ represent a temperature distribution in $m$ locations in Australia. Entries of $\x$ and corresponding locations are represented in Table \ref{table1}.
\begin{table}[h]
\centering
\begin{tabular}{|c|c|}
  \hline
  Entry & Location \\
  \hline
  $\x_1$ & Canberra \\
  \hline
  $\x_2$ & Tuggeranong \\
  \hline
  $\x_3$ & Sydney \\
  \hline
  $\x_4$ & Penrith \\
  \hline
  $\x_5$ & Wollongong \\
  \hline
  $\x_6$ & Melbourne \\
  \hline
  $\x_7$ & Ballarat \\
  \hline
  $\x_8$ & Albury-Wodonga \\
  \hline
  $\x_{9}$ & Bendigo \\
  \hline
  $\x_{10}$ & Brisbane \\
  \hline
  $\x_{11}$ & Cairns \\
  \hline
  $\x_{12}$ & Townsville \\
  \hline
  $\x_{13}$ & Gold Coast \\
  \hline
  $\x_{14}$ & Adelaide \\
  \hline
  $\x_{15}$ &  Mount Gambier \\
  \hline
  $\x_{16}$ & Renmark \\
  \hline
  $\x_{17}$ & Port Lincoln \\
  \hline
\end{tabular}
\begin{tabular}{|c|c|}
  \hline
  Entry & Location \\
  \hline
  $\x_{18}$ & Perth \\
  \hline
  $\x_{19}$ & Kalgoorlie-Boulder \\
  \hline
  $\x_{20}$ & Broome \\
  \hline
  $\x_{21}$ & Hobart \\
  \hline
  $\x_{22}$ & Launceston \\
  \hline
  $\x_{23}$ & Devonport \\
  \hline
  $\x_{24}$ & Darwin \\
  \hline
  $\x_{25}$ & Alice Springs \\
  \hline
  $\x_{26}$ & Tennant Creek \\
  \hline
  $\x_{27}$ & Casey \\
  \hline
  $\x_{28}$ & Davis \\
  \hline
  $\x_{29}$ & Mawson \\
  \hline
  $\x_{30}$ & Macquarie Island \\
  \hline
  $\x_{31}$ & Christmas Island \\
  \hline
  $\x_{32}$ & Cocos Island \\
  \hline
  $\x_{33}$ & Norfolk Island \\
  \hline
  $\x_{34}$ & Howe Island \\
  \hline
\end{tabular}
\caption{\large\small Distribution of entries of signal ${\bf \x}$ and locations in Australia.}
\label{table1}
\end{table}

Suppose $\omega\in\Omega$ is associated with time $t_\omega\in[0, 24]$ of the temperature measurement. Then $\x_j(\omega)$, for $j=1,\ldots,m$, is a temperature in the $j$th location at time $t_\omega$. The values of the minimum and maximum daily temperature, and the temperature at 9am and 3pm, in $m=34$  specific locations of Australia, for each day in 2016,  are provided\footnote{There were 366 days in 2016. There are four values of temperature per day. Thus, we have 1464  temperature values for 2016. From a statistical point of view, the temperature measurements should be provided in more times of a day but such data are not available for us.} by the Bureau of Meteorology of the Australian Government \cite{australia_temp}.
 In particular, the distribution of the  maximum daily temperature in 2016 in all 34 locations in Australia is diagrammatically represented in Fig. \ref{fig3}.
 Let $t_{min}$ and $t_{max}$ denote times when minimal and maximum temperature occur.  We denote by $\omega_{t_{min}}$, $\omega_{t_{max}}$, $\omega_{9{am}}$ and $\omega_{3{pm}}$ outcomes associated with times $t_{min}$, $t_{max}$, 9{am} and 3{pm}, respectively. Further, for $j=1,\ldots,m$, we denote by $\x_{j, (date)}(\omega)$ a temperature in the $j$-th location at time $t_\omega$ on the date labeled as `date'.
For example, on $07/07/2016$, the corresponding temperature values in Ballarat are
$\x_{7, (07/07/2016)}(\omega_{t_{min}}) = 7.4,$  $\x_{7, (07/07/2016)}(\omega_{9 {am}}) = 8.1$,
$\x_{7, (07/07/2016)}(\omega_{3{pm}}) = 9.2,$  $\x_{7, (07/07/2016)}(\omega_{t_{max}}) = 9.5.$

 Let ${\bf y}=A{\bf x}+\bxi$ where $A\in\mathbb{R}^{m\times m}$ is an arbitrary matrix with uniformly distributed random entries and  $\bxi\in L^2(\Omega,\mathbb{R}^{m})$  is white noise, i.e. $E_{\xi \xi} = \sigma^2 I$.
Further, let $\w\in L^2(\Omega,\mathbb{R}^{\ell})$ and $\h\in L^2(\Omega,\mathbb{R}^{\eta})$ be  Gaussian random vectors used in the three-term PCA given by (\ref{sjk49}), (\ref{wn202}).
 It is assumed that noise $\bxi$ is  uncorrelated with ${\bf x}$, $\w$ and ${\bf h}$.
 %Therefore, $E_{xy}=E_{xx}A^T,$ $E_{wy}=E_{wx}A^T$ and  $E_{yy}=AE_{xx}A^T+E_{\xi\xi}.$
Covariance matrices  are represented in terms of samples. For example, $E_{xw} = \frac{1}{p}X W^T$ and $E_{hh} = \frac{1}{p} H H^T$ where $X\in\mathbb{R}^{m\times p}$,  $W\in\mathbb{R}^{\ell\times p}$ and $H\in\mathbb{R}^{\eta\times p}$ are samples of ${\bf x}$,  ${\bf w}$ and $\h$, and  $p$ is the number of samples.  Other covariance matrices are represented similarly. In this example, we consider four  specific  samples of $\x$ as follows. Matrices $X=X_{9am}\in\mathbb{R}^{m\times 366}$ and  $X=X_{3pm}\in\mathbb{R}^{m\times 366}$ represent  the temperature taken each day in 2016 at 9am in all $m$ locations. Entries of matrices $X_{max}\in\mathbb{R}^{m\times 366}$ and $X_{min}\in\mathbb{R}^{m\times 366}$ are values of the maximum and minimum temperature, respectively, for each day in 2016 in all $m$ locations. The database was taken from the Bureau of Meteorology website of Australian Government \cite{australia_temp}.  Matrices $W$ and $H$ were created by MATLAB command {\tt rand(m,p)}.

\begin{figure}[h]
  \centering
  \includegraphics[scale=0.65]{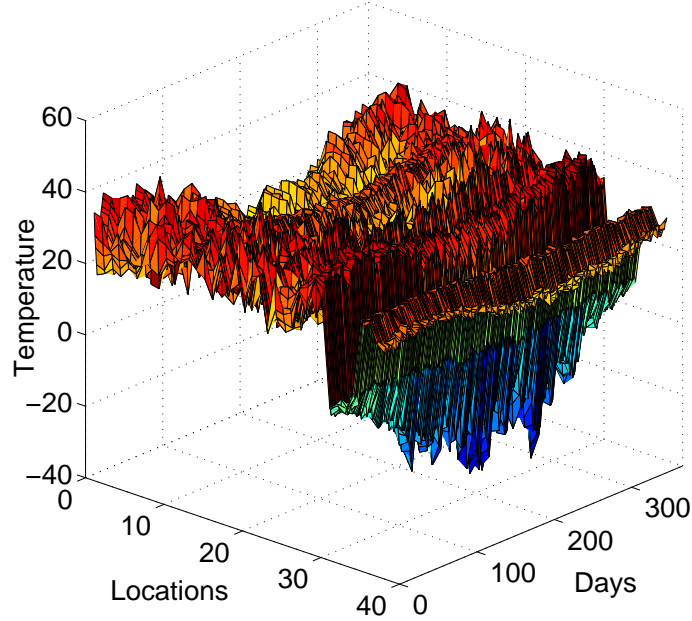}\\
  \caption{Values of maximal temperature in 34 locations of Australia in 2016}
  \label{fig3}
\end{figure}

We  consider eight different cases of simulations. In each case, matrix $X$ is chosen either as one of matrices $X_{9am}$, $X_{3pm}$, $X_{min}$, $X_{max}$, or their combinations as follows:
\begin{table}[h]
\centering
\begin{tabular}{|c|c|c|c|c|c|}
  \cline{2-6}
  \multicolumn{1}{c|}{} & {\small Case 1} & {\small Case 2} & {\small Case 3}  & {\small Case 4} & {\small Case 5} \\
  \hline
 $X$  & $X_{min}$ & $X_{9am}$ & $X_{3pm}$  & $X_{max}$ & $[X_{9am}, X_{3pm}]$ \\
  \hline
$p$  & $366$ & $366$ & $366$  & $366$ & $732$ \\
  \hline
\end{tabular}

\vspace{0.7cm}

\begin{tabular}{|c|c|c|c|}
  \cline{2-4}
  \multicolumn{1}{c|}{}& {\small Case 6}  & {\small Case 7} & {\small Case 8} \\
  \hline
$X$  & $[X_{min}, X_{max}]$ & $[X_{min}, X_{9am}, X_{3pm}]$ & $[X_{min}, X_{9am}, X_{3pm}, X_{max}]$ \\
  \hline
$p$ & $732$ & $1098$ & $1464$  \\
 \hline
\end{tabular}
\end{table}
For each case, the diagrams of the  error associated with the GBT1, GBT2 and three-terms PCA in form (\ref{sjk49})-(\ref{wn202}), for  $m=\ell = 34$, $k=17$ and $\sigma=1$, versus the dimension $\eta=0, 1, \ldots, 500$ of vector ${\bf h}$ are shown in Fig. \ref{fig4}.

%\vspace*{-5mm}
\begin{figure}[h!]
    \centering
    \begin{tabular}{cc}
      \includegraphics[scale=0.44]{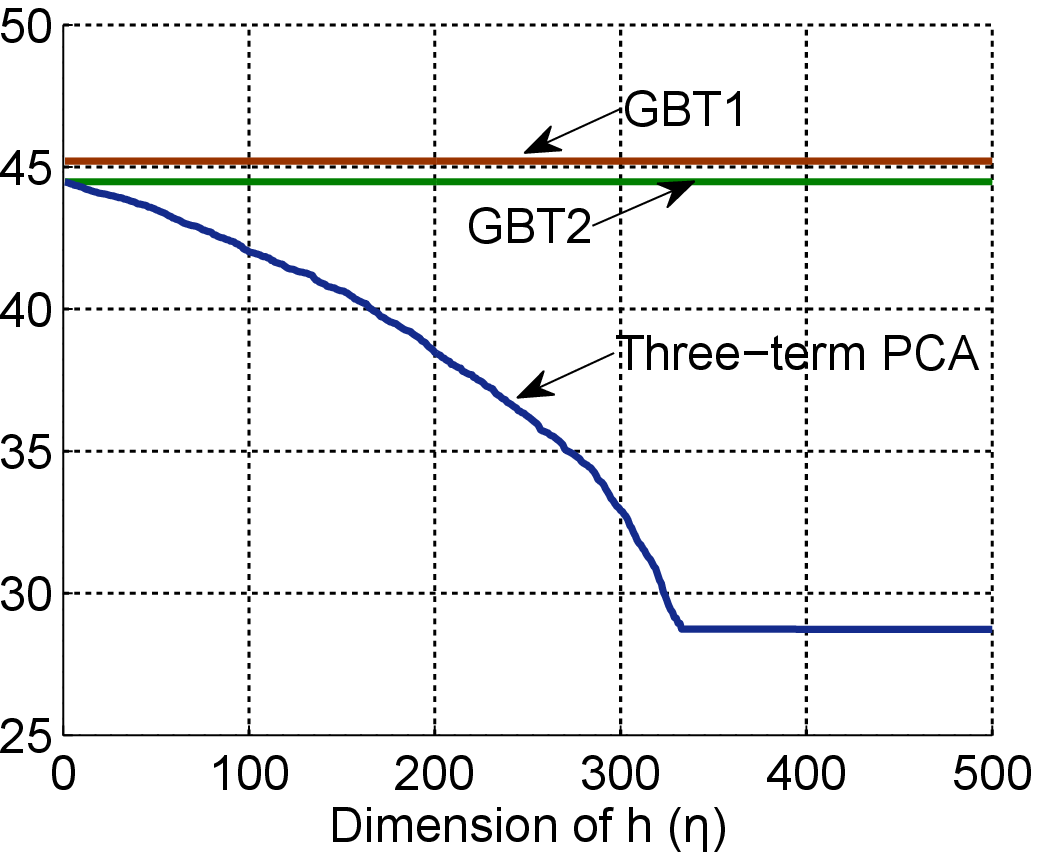}  & \includegraphics[scale=0.44]{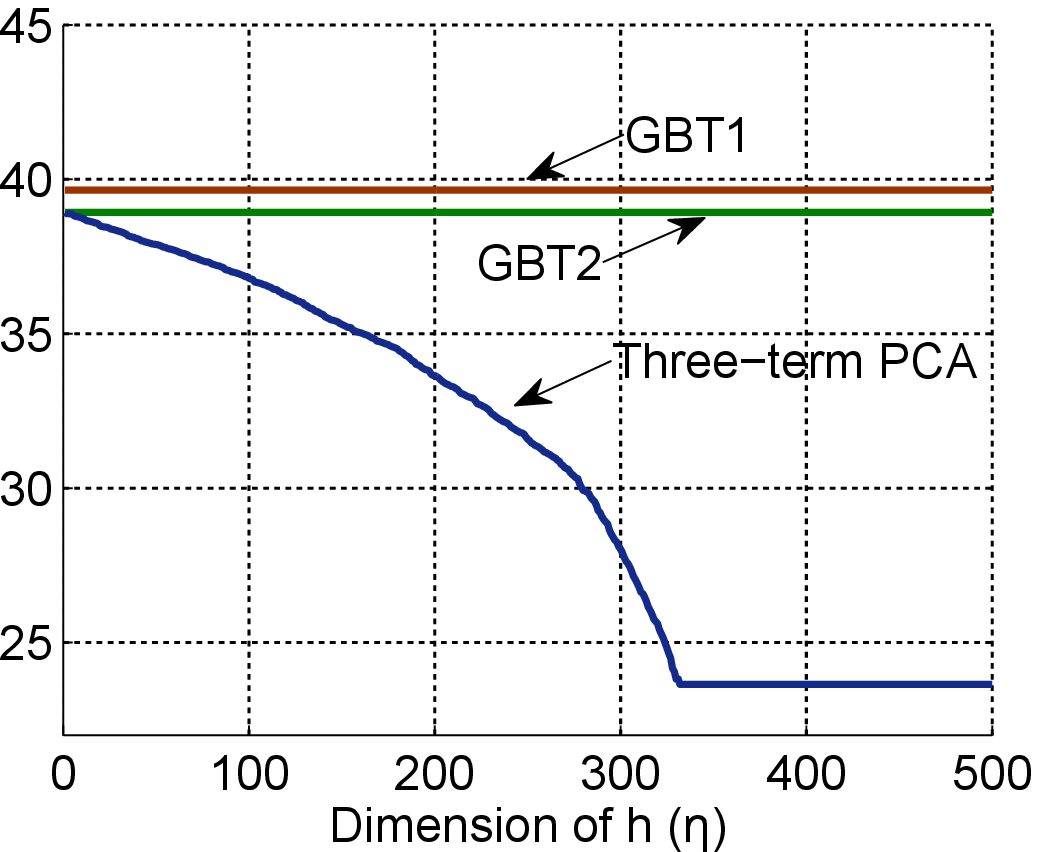}\\
      {\small (a) Case 1: $X=X_{min}$} & {\small (b) Case 2: $X=X_{9am}$} \\
      \includegraphics[scale=0.44]{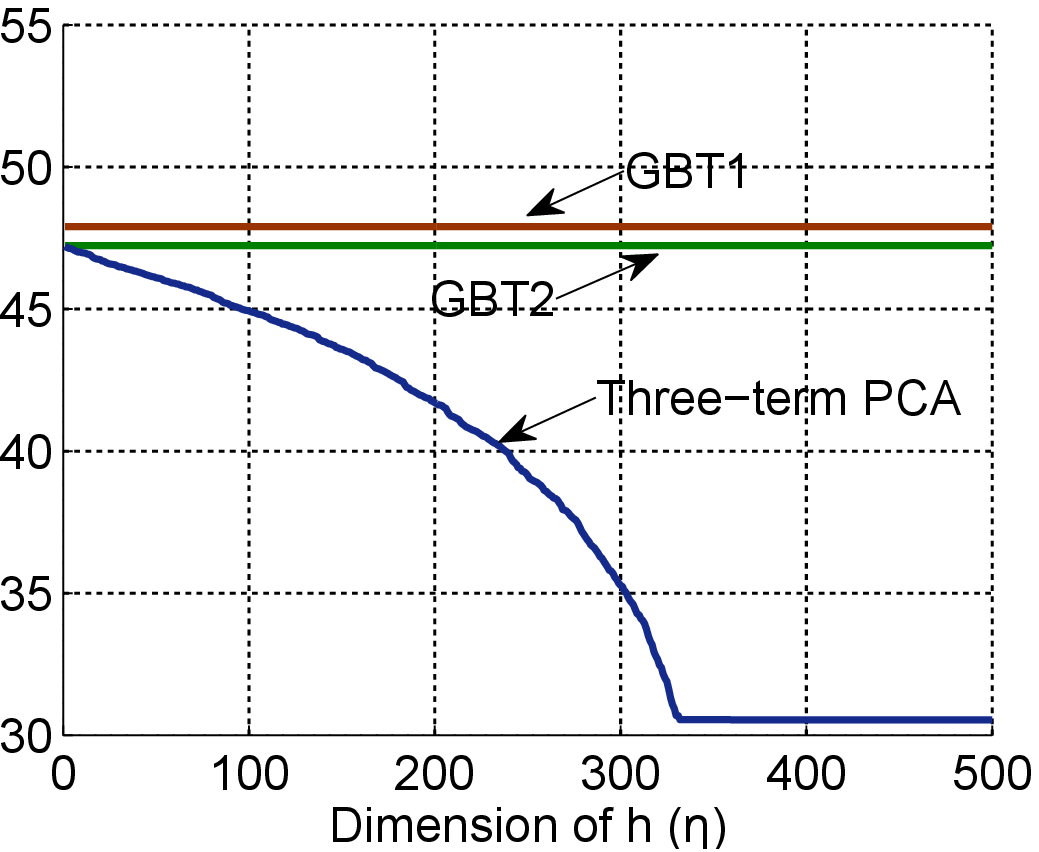}  & \includegraphics[scale=0.44]{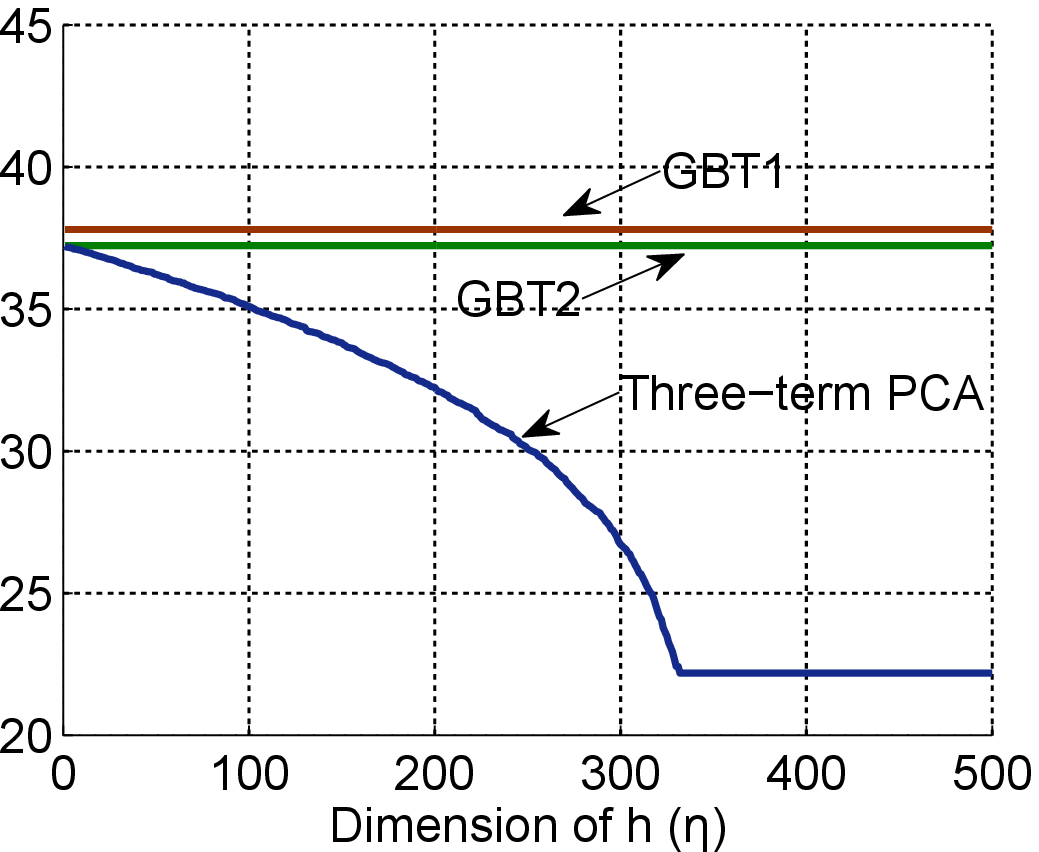}\\
      {\small (c) Case 3: $X=X_{3pm}$} & {\small (d) Case 4: $X=X_{max}$} \\
      \includegraphics[scale=0.44]{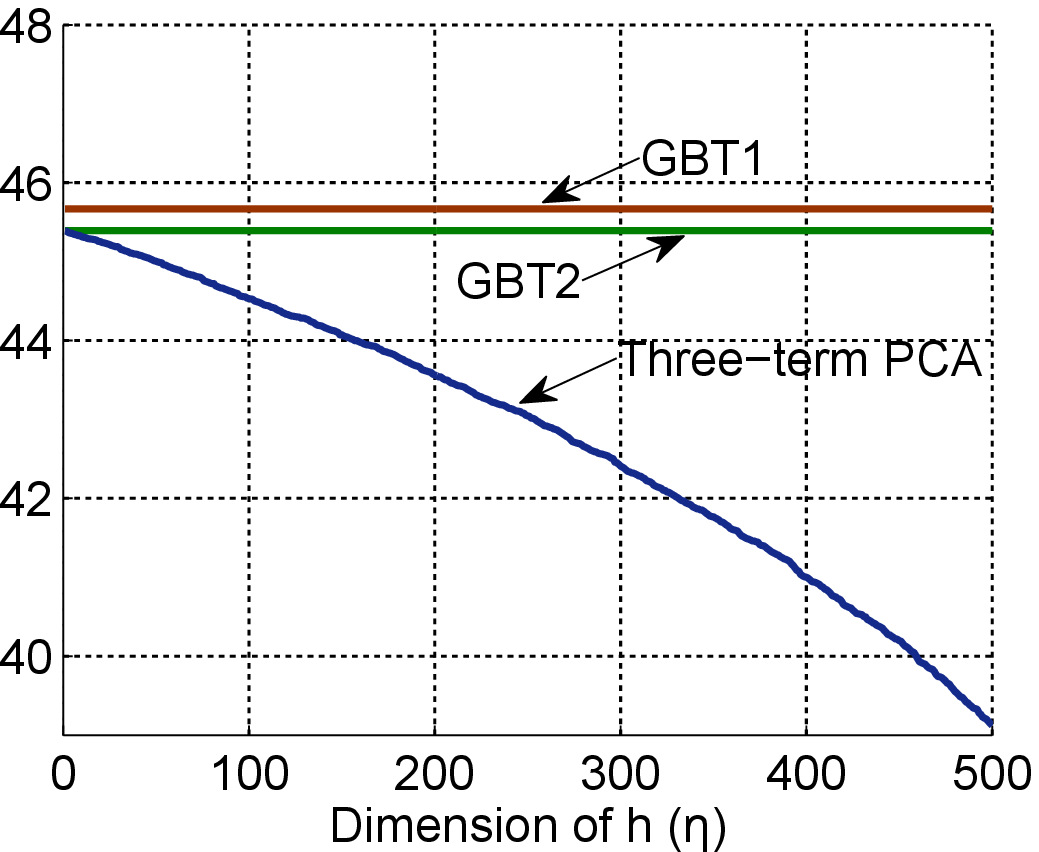} & \includegraphics[scale=0.44]{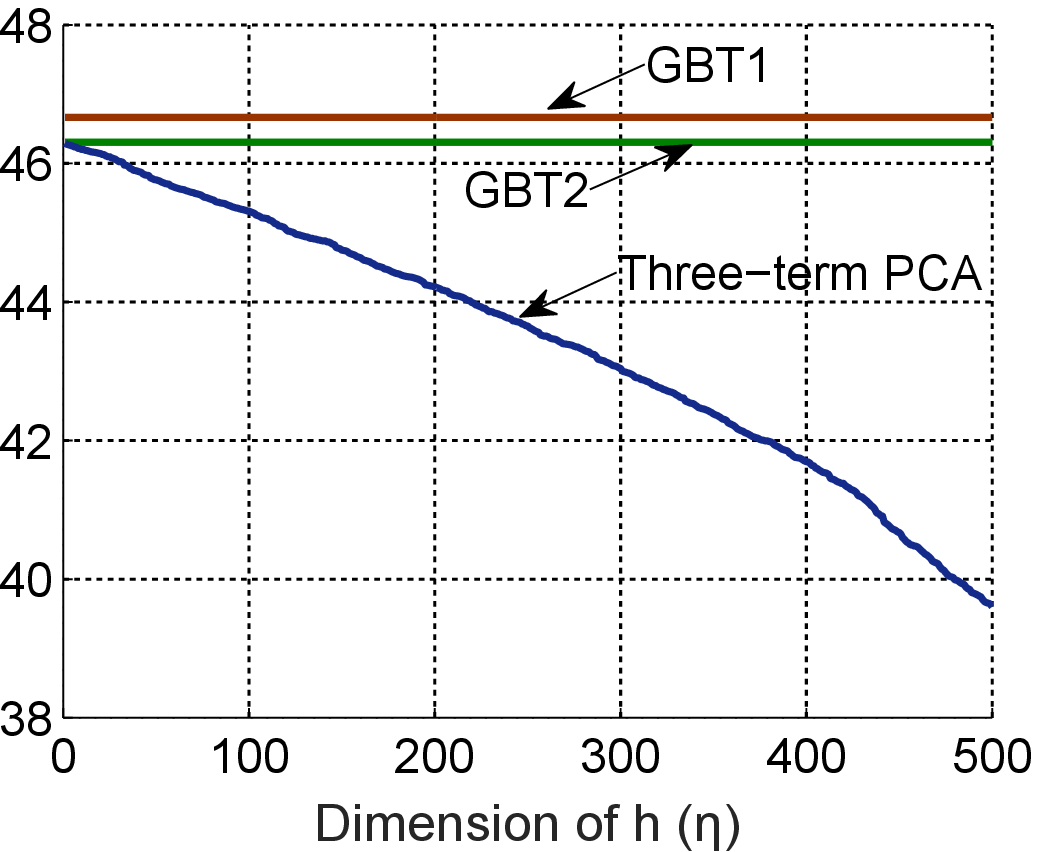}\\
      {\small (e) Case 5: $X=[X_{9am}\;X_{3pm}]$} & {\small (f) Case 6: $X=[X_{min}\;X_{max}]$} \\
      \includegraphics[scale=0.44]{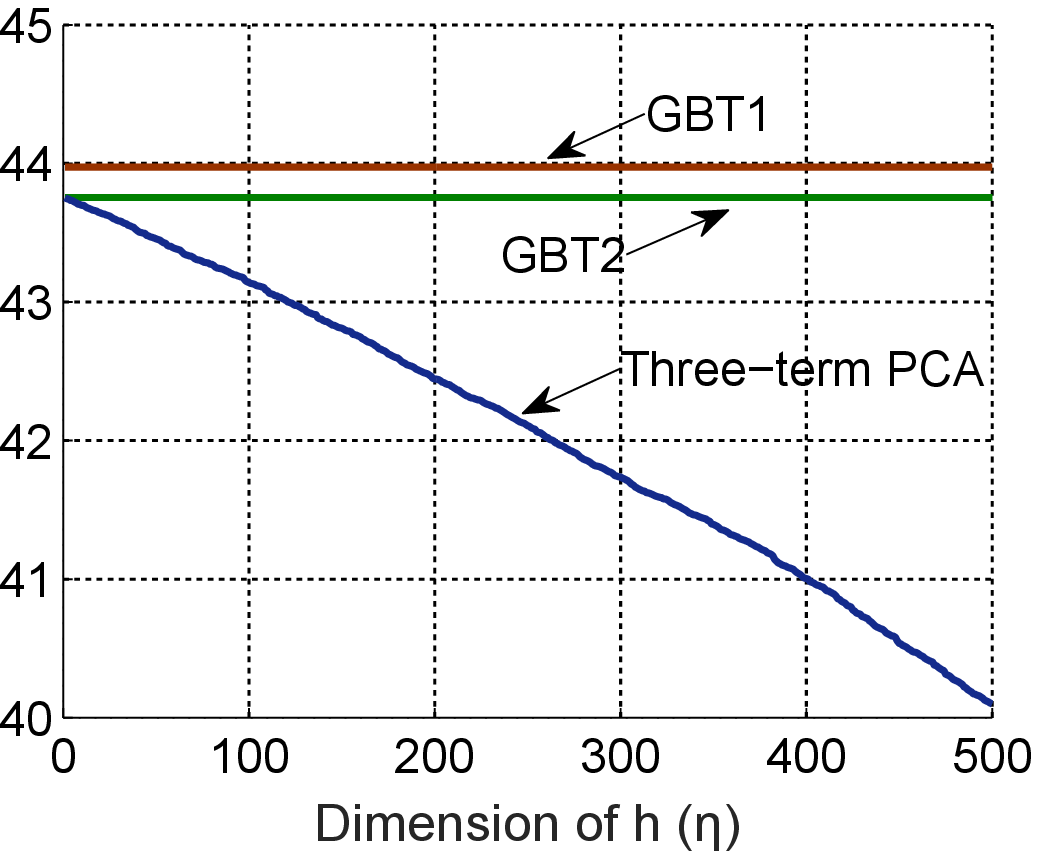}  & \includegraphics[scale=0.44]{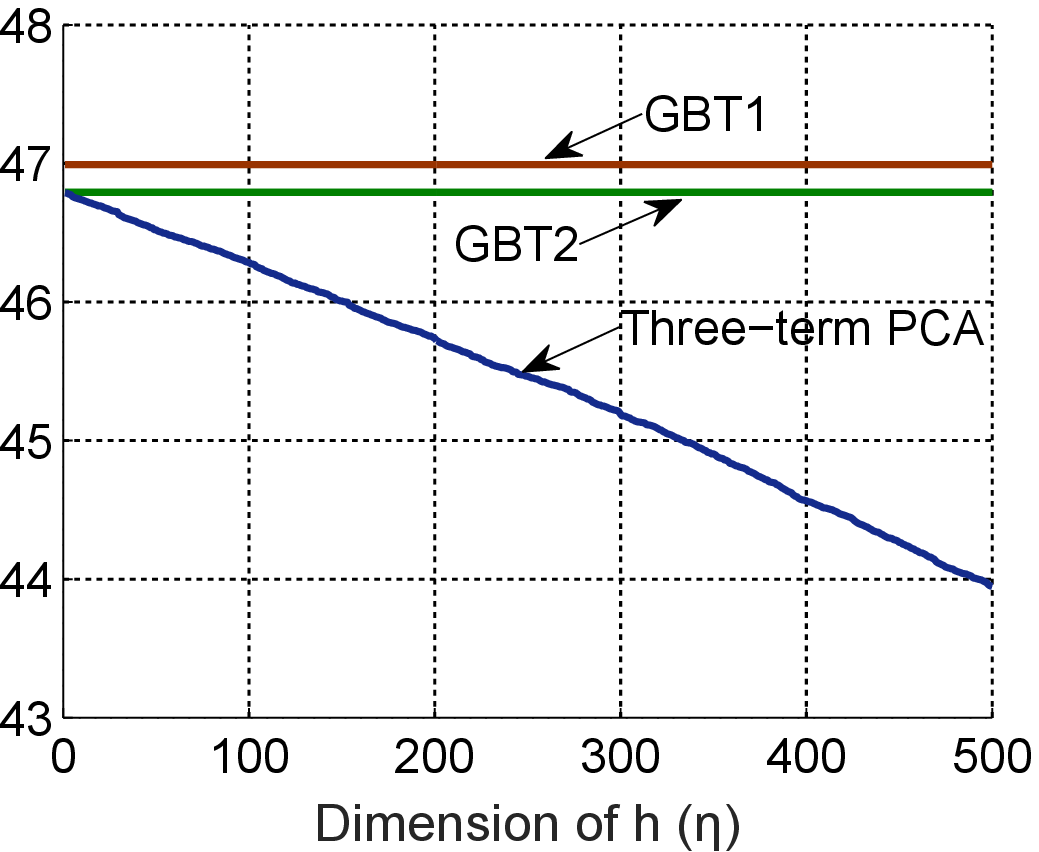}\\
      {\small (g) Case 7: $X=[X_{min}\;X_{9am}\;X_{3pm}]$} & {\small (h) Case 8: $X=[X_{min\;}X_{9AM}\;X_{3PM}\;X_{max}]$}\\
    \end{tabular}
    \caption{Errors associated with the GBT1, GBT2 and three-terms PCA versus dimension $\eta$ of $H$.}
    \label{fig4}
\end{figure}

It follows from the diagrams for all cases represented in  Fig. \ref{fig4}, that the error associated with the three-term PCA  decreases as $\eta$ increases. This is the numerical illustration of Theorem \ref{w791n} and the Case 3 considered in Section \ref{vnb01m}. In particular, in Figs. \ref{fig4} (a)-(d), for $\eta= 335,\ldots, 500$, the error associated with the three-term PCA coincides with that of the GBT1 applied to the case when $\y = \x$, i.e.. with that of PCA applied to the data {\sf\em  without any noise.}

Further, Fig. \ref{fig5} illustrates the following  observation. For  $\sigma\in [0, 2]$ and  $\eta\in [0, 300]$, the behavior of the error associated with the three-term PCA  is similar:  the increase in $\eta$ implies the decrease in the error. Interestingly, for $\eta\in [300, 500]$, the error remains  constantly small regardless of  the value of $\sigma$. Similar to the above, for $\eta\in [300, 500]$, the error associated with the three-term PCA coincides with that that of PCA applied to the data {\sf\em without any noise.}

\begin{figure}[h!]
  \centering
  \includegraphics[scale=0.6]{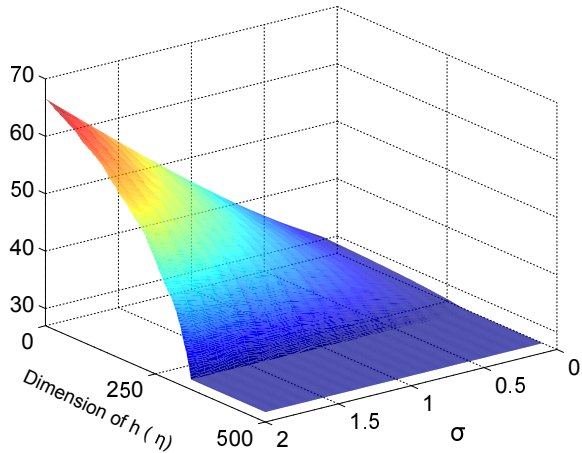}\\
  \caption{Diagrams of errors associated with the  three-terms PCA versus dimension $\eta$ of $\h$ and values of $\sigma$ of   $E_{\xi\xi}$.}
  \label{fig5}
\end{figure}
\end{example}

\subsection{ Case 4.   Pure filtering}\label{}

One more special case of the three-term PCA is as follows.
Consider transform  $\ttt_1$  defined by
\begin{eqnarray}\label{sj92}
\ttt_1(\y,\w) =  A_0\y + A_1 \f(\y, \w),
\end{eqnarray}
where $A_0\in\rt^{m\times n}$ and $A_1\in\rt^{m\times \ell}$ are  full rank matrices. Optimal $A_0$ and $A_1$ are determined from the solution of the following problem:
  Given  $E_{xy}$, $E_{yy}$, $E_{yw}$ are $E_{ww}$, find  full rank $A_0$ and $A_1$  that solve
\begin{eqnarray}\label{zbn91}
\min_{{A_0, A_1}} \|\x - [A_0\y + A_1 \f(\y, \w)]\|^2_\Omega.
\end{eqnarray}
In other words, $\ttt_1$ is a pure filter, with no principal component determination.
As before, we write $\s=\f(\y, \w)=\w- E_{wy} E_{yy}^\dag \y$. We call $\ttt_1$ the three-term filter (TTF).

\begin{theorem}\label{1nm0}
Minimal norm solution to problem (\ref{zbn91}) is given by
\begin{eqnarray}\label{68bn1}
 A_0=E_{x y} E_{yy}^\dag \qa A_1=E_{x s} E_{ss}^\dag.
\end{eqnarray}
The associated error is represented by
\begin{eqnarray}\label{xmabn1}
&&\hspace*{-10mm}\min_{A_0, A_1}  \|\x - [A_0 \y + A_1 \s]\|^2_\Omega\nonumber\\
&& \hspace*{-3mm}= \|E_{xx}^{1/2}\|^2 -  \|[E_{xy}{E_{yy}^{1/2}}^\dag\|^{2} -  \|E_{xs}{E_{ss}^{1/2}}^\dag ]\|^{2}.
\end{eqnarray}
\end{theorem}

\begin{proof}
Optimal full rank  $A_0$ and $A_1$ given by (\ref{68bn1}) follow from (\ref{xm02}) where ${E_{zz}^\dag}^{1/2}$ is represented by ${E_{zz}^\dag}^{1/2} = \left[ \begin{array}{cc}
                         {E_{yy}^\dag}^{1/2}  & \oo\\
                         \oo  & {E_{ss}^\dag}^{1/2}
                           \end{array} \right]$ and $[E_{xz}{E_{zz}^\dag}^{1/2}]_{k}$   should be replaced with $E_{xz}{E_{zz}^\dag}^{1/2}$.
Further, for $A=[A_0\hspace*{1mm} A_1]$,
\begin{eqnarray}\label{}
 \|\x - [A_1 \y + A_1 \s]\|^2_\Omega = \|E_{xx}^{1/2}\|^2 -   \|E_{xz}{E_{zz}^{1/2}}^\dag\|^{2} + \|E_{xz}{E_{zz}^{1/2}}^\dag - A{E_{zz}^{1/2}}\|^{2}.
\end{eqnarray}
For $A_0$ and $A_1$ given by (\ref{68bn1}), $A=[E_{x y} E_{yy}^\dag \hspace*{1mm}E_{x s} E_{ss}^\dag]=E_{xz}E_{zz}^\dag$. Therefore,
\begin{eqnarray*}
\hspace*{-17mm}\min_{A_1, A_1}  \|\x - [A_1 \y + A_1 \s]\|^2_\Omega \hspace*{-5mm}& & =\|E_{xx}^{1/2}\|^2 -   \|E_{xz}{E_{zz}^{1/2}}^\dag\|^{2} \\
                                                     &  &\hspace*{10mm}+ \|E_{xz}{E_{zz}^{1/2}}^\dag - E_{xz}E_{zz}^\dag {E_{zz}^{1/2}}\|^{2}\\
                                                     && \hspace*{20mm} =\|E_{xx}^{1/2}\|^2 -   \|E_{xz}{E_{zz}^{1/2}}^\dag\|^{2}
\end{eqnarray*}
because  $E_{zz}^\dag {E_{zz}^{1/2}} = {E_{zz}^{1/2}}^\dag$ {\em \cite[p. 313]{torbook2007}}. Then (\ref{xmabn1}) follows.\hfill$\blacksquare$
\end{proof}

The accuracy of the optimal TTF $\ttt_1$ is better than that of the optimal linear filter $\widetilde{F}=A_0= E_{x y} E_{yy}^\dag$ \cite{torbook2007}. More specifically, the following is true.

\begin{theorem}\label{nm012}
The error associated with the optimal TTF $\ttt_1$ is less than that of the optimal linear filter $\widetilde{F}=A_1$ by $\|E_{xs}{E_{ss}^{1/2}}^\dag ]\|^{2}$, i.e.,
\begin{eqnarray}\label{wom9}
\min_{A_0, A_1} \|\x - [A_0 \y + A_1 \f(\y, \w)]\|^2_\Omega = \min_{A_0} \|\x - A_0 \y \|^2 -  \|E_{xs}{E_{ss}^{1/2}}^\dag ]\|^{2}_\Omega.
\end{eqnarray}
\end{theorem}

\begin{proof}
The proof follows directly from (\ref{xmabn1}).  \hfill$\blacksquare$
\end{proof}

%Theorem \ref{tui9} and Remark \ref{hjk29} are illustrated with the following example.
%%%%%%%%%%%%%%%%%%%%%%%%%%%%%%%%%%%%%%%

\section{Advantages of three-term PCA (continued)}\label{x78kl}

%\subsection{Comparison of three-term PCA with GBT2 in \cite{tor843}}\label{mcb78i}

\subsection{Increase in accuracy compared to that of GBT2 \cite{tor843} and GKLT}\label{58b29}
% \begin{eqnarray}\label{sjk92}
% \widetilde{\x}=\ttt (\y, \w) = D_0 C_0\y +  D_1 C_1 H(\y, \w),
% \end{eqnarray}
% where
% $$
% D_0 = D_1=U_{G_z,k}, \quad C_0 =U_{G_z,k}^T G_{xy} \qa C_1 = U_{G_z,k}^T  G_{xs}
% $$

The three-term PCA  represented by (\ref{sjk92}) and (\ref{sjk49})  has more parameters to optimize than those in the GBT2 (\ref{bnm11}) and GKLT (\ref{qmiw8}), i.e., it has more degrees of freedom to control the performance. Indeed, in the three-term PCA, the dimension of $\w$-injection and $\h$-injection are $\ell$ and $\eta$, and they can be varied while in the GBT2, the dimension of auxiliary vector $\vv$ is $n$ and it is fixed. In the GKLT dimension of vector $\y^2$ is also fixed.
 Thus, in the three-term PCA, $\ell$ and $\eta$ represent the additional degrees of freedom.

 By Theorem \ref{w791n}, the increase in $\eta$ implies the improvement in the accuracy of the three-term PCA. Thus, unlike the GBT2, the performance of the three-term PCA  is improved because of  the increase in the dimension of $\h$-injection.

\subsection{Improvement in the associated numerical load}\label{cm9vn}

In a number of applied problems, dimensions $m,n$ of associated covariance matrices are large. For instance, in the DNA array analysis \cite{Alter11828,ziyang18},  $m={\it O}(10^4)$. In this case, the associated numerical load needed to compute the covariance matrices increases significantly. Therefore, a method which requires a lower associated numerical load is, of course, preferable.

Here, we wish to illustrate the computational advantage of the three-term PCA represented by (\ref{sjk92}),  (\ref{smi91}) and  (\ref{294b}) compared to that of the GBT2 in (\ref{bnm11}) and the GKLT \cite{924074}. In both methods, a pseudo-inverse matrix is evaluated by the SVD.
The computational load  of the  three-term PCA (abbreviated as $C_{PCA3}$) consists of the matrix products in (\ref{294b}), and computation of SVDs for $E_{yy}^\dag\in\rt^{n\times n}$, $E_{ss}^\dag\in\rt^{\ell\times \ell}$ and $G_z\in\rt^{m\times m}$. Recall that $E_{zz}^\dag = \left[ \begin{array}{cc}
                         E_{yy}^\dag  & \oo\\
                         \oo  & E_{ss}^\dag
                           \end{array} \right]$.
The GBT2 computational load ($C_{GBT2}$) consists of computation of the matrix products   in (\ref{ac11}) and  the SVD for  $E_{qq}^\dag\in\rt^{2n\times 2n}$. The computational load of the GKLT ($C_{GKLT}$) contains computation of the matrix products given in \cite{924074}, and computation of the SVDs for $2n\times 2n$ and $m\times 2n$ matrices.
Importantly, the dimensions of matrices in the  three-term PCA are less than those in the GBT2 and GKLT. For example, the dimension of $E_{yy}$ is twice less than that of $E_{qq}$. This circumstance implies  the decrease in the computational load of the   three-term PCA compared to that in the GBT2 and GKLT. Indeed, the product of $m\times n$ and $n\times p$ matrices requires approximately $2mnp$ flops, for large $n$. The Golub-Reinsch SVD method (as given in \cite{golub1996}, p. 254) requires $4mn^2 + 8n^3$ flops to compute the pseudo-inverse for a $m\times n$ matrix\footnote{The Golub-Reinsch SVD method appears to be more effective than other related methods considered in \cite{golub1996}.}. As a result, for $m=n=\ell$,
\begin{eqnarray}\label{29cn3}
%C_{PCA3} = 32m^3 + 8m^2 k
C_{PCA3} = 52m^3 + 2m^2(k+1)
\end{eqnarray}
while
\begin{eqnarray}\label{88cn3}
C_{GBT2} = 140m^3 + 2m^2(k+2) \mbox{ and } C_{GKLT} = 240m^3 + 4m^2(k+1)+mk .
\end{eqnarray}
That is, for large $m$ and $n$, $C_{PCA3}$ is about $37\%$ of  $C_{GBT2}$ and  $22\%$ of  $C_{GKLT}$ .
%Interestingly, the computational load of the GBT1\footnote{Recall the GBT1 is a particular case of the GBT2 (\ref{bnm11}) if $R_1 P_2\vv=\mathbf 0$.} $C_{GBT1}$, for $m=n$, is given by
%\begin{eqnarray}\label{00cn3}
%C_{GBT1} = 52m^3 + 2m^2 k
%\end{eqnarray}

Thus, the three-term PCA may provide a better associated accuracy than that of the GBT2 and GKLT (see Section \ref{1n8an},  Example \ref{m2b9} and Section \ref{58b29}) under the computational load which, for large $m, n$, is about one third of $C_{GBT2}$ and a quarter of $C_{GKLT}$.
This observation is illustrated by the following numerical example.

\begin{example}\label{nm498}
Let ${\bf y}=A{\bf x}+\xi$ where ${\bf x}\in L^2(\Omega,\mathbb{R}^{m})$ is a uniformly distributed random vector, $\xi\in L^2(\Omega,\mathbb{R}^{m})$ is a Gaussian random vector with variance one and $A\in\mathbb{R}^{m\times m}$ is a matrix with normally distributed random entries. We choose  ${\bf w}\in L^2(\Omega,\mathbb{R}^{\ell})$ as an  uniformly distributed random vector.
Covariance matrices $E_{xx}$, $E_{ww}$ and $E_{\xi\xi}$ are represented by $E_{yy}=\frac{1}{s} YY^T,\quad E_{ww}=\frac{1}{s}WW^T, \quad E_{\xi\xi}=\sigma^2 I,$ where  $Y\in\mathbb{R}^{m\times p}$ and $W\in\mathbb{R}^{m\times p}$ are corresponding sample matrices, $p$ is a number of samples, and  $\sigma=1$. Suppose only samples of $\y$ and $\w$ are available, and for simplicity let us assume that matrix $A$ is invertible. Then, in particular,
$
E_{xx} = A^{-1}(E_{yy} - E_{\xi\xi})A^{-T}\qa    E_{xy}=E_{xx} A^T.
$
\begin{figure}[h!]
  \centering
  \includegraphics[scale=0.60]{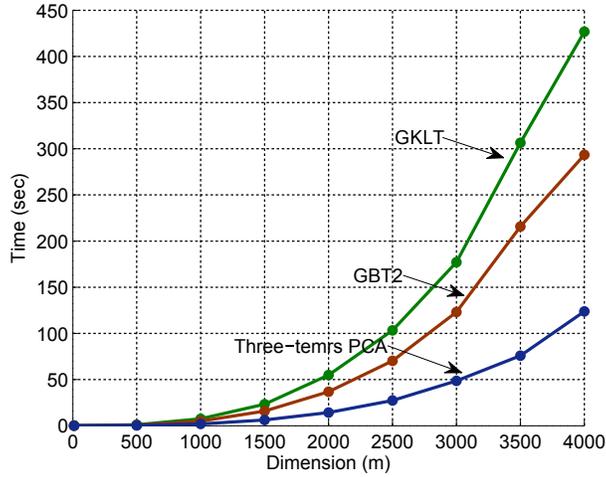}\\
  \caption{Example \ref{nm498}: Time versus matrix dimension $m$ used to execute the three-term PCA  (blue line), GBT2 (red line) and GKLT (green line). }
  \label{fig2}
\end{figure}

In Fig. \ref{fig2}, for a randomly chosen $A$, and $m=\ell$ and $p=3m$,  typical diagrams of time (in sec.) used to execute the three-term PCA in form (\ref{sjk92}) and (\ref{294b}), GBT2   (\ref{bnm11})-(\ref{byfx}) and GKLT \cite{924074} versus  dimension $m$  are represented. The diagrams in Fig. \ref{fig2} confirm the observation made before this example, i.e.,
for large $m$, $C_{PCA3}$  is significantly less than  $C_{GBT2}$  and $C_{GKLT}$ (see (\ref{29cn3}) and  (\ref{88cn3}), respectively).

%Interestingly, in this example, time  used to execute the three-term PCA and  GBT1 is very close.
\end{example}

\section{Final discussion}\label{nbm198}

We have developed the extension of PCA called the three-term PCA. The three-term PCA is presented  in two forms considered in Sections \ref{w0an}, \ref{det1},  and \ref{xvb8an}. The associated advantages have been detailed in Sections \ref{x78an}  and \ref{x78kl}.
We would like to highlight the following observation. The proposed three-term PCA is applied to observed data represented by random vector $\y$. Here, $\y$ is a noisy version of original data $\x$. It is shown that in this case, the error associated with the three-term PCA, $\varepsilon_{m,n,\ell+\eta} (S_0, S_1) $,   is less than that of the known generalizations of PCA, i.e. the GBT1, GKLT and GBT2 (see Section \ref{1n8an}).  At the same time, in the ideal case of the observed data {\em without any noise}, i.e. when $\y=\x$, the three-term PCA coincides with the GBT1 (with $\y=\x$ in the GBT1 as well) which is a generalization of PCA for the case of singular data. Its associated error is minimal among all transforms of the same rank and cannot be improved.
  Let us denote this error by $\varepsilon_{(y=x)} $.
 Example \ref{m2b9} above has discovered an important and quite unanticipated feature of the proposed technique as follows. As dimension $\eta$ of $\h$-injection increases, the error $\varepsilon_{m,n,\ell+\eta} (S_0, S_1) $ decreases  up to  $\varepsilon_{(y=x)} $. This implies a conjecture    that this is true in general, i.e.
 $
 \lim_{\eta \rightarrow \infty} \varepsilon_{m,n,\ell+\eta} (S_0, S_1) = \varepsilon_{(y=x)}.
 $
 We intend to develop a justification of this  observation.

%\section*{References}
\bibliographystyle{plain}

\bibliography{main}

\end{document}